\documentclass{lmcs}
\pdfoutput=1

\usepackage{lastpage}
\lmcsdoi{20}{1}{3}
\lmcsheading{}{\pageref{LastPage}}{}{}%
{Oct.~17,~2022}{Jan.~11,~2024}{}

\title[A Bit of Nondeterminism in Pushdown Automata]{A Bit of Nondeterminism Makes Pushdown\texorpdfstring{\\}{ }Automata Expressive and Succinct}

\thanks{This work is supported by the Indian Science and Engineering Research Board 
(SERB) grant SRG/2021/000466, by the EIPHI Graduate School (contract ANR-17-EURE-0002), and by DIREC – Digital Research Centre Denmark.}

\author[S.~Guha]{Shibashis Guha\lmcsorcid{0000-0002-9814-6651}}[a]
\author[I.~Jecker]{Isma\"el Jecker}[b,c]
\author[K.~Lehtinen]{Karoliina Lehtinen\lmcsorcid{0000-0003-1171-8790}}[d]
\author[M.~Zimmermann]{Martin Zimmermann\lmcsorcid{0000-0002-8038-2453}}[e]

\address{Tata Institute of Fundamental Research, Mumbai, India}	

\address{University of Warsaw, Warsaw, Poland}	

\address{Université de Franche-Comté, CNRS, FEMTO-ST, Besançon, France}

\address{CNRS, Aix-Marseille University, LIS,
Marseille, France}

\address{Aalborg University, Aalborg, Denmark}

\usepackage[utf8]{inputenc}

\usepackage{amsmath}
\usepackage{amssymb}
\usepackage{amsthm}

\usepackage{xspace}

\usepackage{tikz}
\usetikzlibrary{decorations.pathreplacing}
\usetikzlibrary{automata}
\tikzset{>=stealth, shorten >=1pt}
\tikzset{every edge/.style = {thick, ->, draw}}
\tikzset{every loop/.style = {thick, ->, draw}}

\newcommand{\set}[1]{\{#1\}}
\newcommand{\nats}{\mathbb{N}}

\renewcommand{\epsilon}{\varepsilon}
\newcommand{\size}[1]{|#1|}

\newcommand{\gfgpda}{HD-PDA}
\newcommand{\pda}{PDA}

\newcommand{\aut}{\mathcal{P}}
\newcommand{\daut}{\mathcal{D}}

\newcommand{\auta}{\mathcal{A}}

\newcommand{\taut}{\mathcal{T}}

\newcommand{\initmark}{I}

\newcommand{\Gammabot}{\Gamma_{\!\bot}}

\newcommand{\sh}{\mathrm{sh}}

\newcommand{\trans}[1]{\xrightarrow{#1}}

\newcommand{\hstrat}{r}

\newcommand{\lmax}{B_{3}}

\newcommand{\sink}{\texttt{end}}
\newcommand{\strat}{\sigma}

\newcommand{\cfl}{CFL\xspace}
\newcommand{\dcfl}{DCFL\xspace}
\newcommand{\gfgcfl}{HD-CFL\xspace}
\newcommand{\vpl}{VPL\xspace}

\newcommand{\HD}{HD\xspace}
\newcommand{\HDness}{HDness\xspace}

\newcommand{\hd}{history-deterministic}
\newcommand{\hdsm}{history-determinism}

\newcommand{\Sigmacall}{\Sigma_c}
\newcommand{\Sigmaskip}{\Sigma_s}
\newcommand{\Sigmareturn}{\Sigma_r}

\newcommand{\exptime}{\textsc{ExpTime}}

\newcommand{\twoexptime}{\textsc{2ExpTime}}

\newcommand{\good}{{\sc ge}} 

\renewcommand{\ge}{\geqslant}
\renewcommand{\geq}{\geqslant}
\renewcommand{\le}{\leqslant}
\renewcommand{\leq}{\leqslant}

\newcommand{\markingn}[2]{n_{#1,#2}}
\newcommand{\markingl}[2]{\ell_{#1,#2}}
\newcommand{\markingr}[2]{r_{#1,#2}}

\begin{document}


\begin{abstract}
We study the expressiveness and succinctness of history-deterministic pushdown automata (HD-PDA) over finite words, that is, pushdown automata whose nondeterminism can be resolved based on the run constructed so far, but independently of the remainder of the input word. These are also known as good-for-games pushdown automata.

We prove that HD-PDA recognise more languages than deterministic PDA (DPDA) but not all context-free languages (CFL). This class is orthogonal to unambiguous CFL.
We further show that HD-PDA can be exponentially more succinct than DPDA, while PDA can be double-exponentially more succinct than HD-PDA.
We also study HDness in visibly pushdown automata (VPA), which enjoy better closure properties than PDA, and for which we show that deciding HDness is {\sc ExpTime}-complete.
HD-VPA can be exponentially more succinct than deterministic VPA, while VPA can be exponentially more succinct than HD-VPA.
Both of these lower bounds are tight.

We then compare \HD-PDA with PDA for which composition with games is well-behaved, i.e. \textit{good-for-games} automata. We show that these two notions coincide, but only if we consider potentially infinitely branching games.

Finally, we study the complexity of resolving  nondeterminism in HD-PDA.
Every HD-PDA has a positional resolver, a function that resolves nondeterminism and that is only dependent on the current configuration.
Pushdown transducers are sufficient to implement the resolvers of HD-VPA, but not those of HD-PDA.
HD-PDA with finite-state resolvers are determinisable.
\end{abstract}

\maketitle

\section{Introduction}
\label{sec_intro}

Nondeterminism adds both expressiveness and succinctness to deterministic pushdown automata. Indeed,
the class of context-free languages (CFL), recognised by nondeterministic pushdown automata (PDA), is strictly larger than the class of deterministic context-free languages (DCFL), recognised by deterministic pushdown automata (DPDA), both over finite and infinite words.
Even when restricted to languages in DCFL, there is no computable bound on the relative succinctness of PDA~\cite{Hartmanis80,Valiant76}.
In other words, nondeterminism is remarkably powerful, even for representing deterministic languages. 
The cost of such succinct representations is algorithmic: problems such as universality and solving games with a CFL winning condition are undecidable for PDA~\cite{DBLP:journals/tcs/Finkel01a,Hopcroft}, while they are decidable for DPDA~\cite{DBLP:journals/iandc/Walukiewicz01}.
Intermediate forms of automata that lie between deterministic and nondeterministic models have the potential to mitigate some of the disadvantages of fully nondeterministic automata while retaining some of the benefits of the deterministic ones. 

Unambiguity and bounded ambiguity, for example, restrict nondeterminism by requiring words to have at most one or at most $k$, for some fixed $k$, accepting runs. Holzer and Kutrib survey the noncomputable succinctness gaps between unambiguous PDA and both PDA and DPDA~\cite{HK10}, while Okhotin and Salomaa show  that unambiguous visibly pushdown automata are exponentially more succinct than DPDA~\cite{okhotin2015descriptional}. 
Universality of unambiguous PDA is decidable, as it is decidable for unambiguous context-free grammars~\cite{DBLP:books/daglib/0067812}, which are effectively equivalent~\cite{Her97}.
However, to the best of our knowledge, unambiguity is not known to reduce the algorithmic complexity of solving games with a context-free winning condition.

Another important type of restricted nondeterminism that is known to reduce the complexity of universality and solving games has been studied under the names of good-for-games nondeterminism~\cite{HP06} and history-determinism~\cite{Col09} (HD)\footnote{After the publication of the conference version~\cite{GJLZ21}, which uses the term good-for-games, it has come to light that the notions of history-determinism and good-for-gameness do not always coincide~\cite{BL21}. We therefore prefer to use the term history-determinism as it corresponds better to our definitions. In Section~\ref{sec:composition}, we show that at least for pushdown automata over \textit{finite} words and potentially infinitely branching games, the two notions coincide.}. Intuitively, a nondeterministic automaton is \HD if its nondeterminism can be resolved on-the-fly, i.e.\ without knowledge of the remainder of the input word to be processed.

    For finite automata on finite words, where nondeterminism adds succinctness, but not expressiveness, \HD nondeterminism does not even add succinctness: every \HD-NFA contains an equivalent DFA~\cite{BKS17}, which can be obtained by pruning transitions from the \HD-NFA.
    Thus, \HD-NFA cannot be more succinct than DFA.
    But for finite automata on infinite words, where nondeterminism again only adds succinctness, but not expressiveness, \HD coBüchi automata can be exponentially more succinct than deterministic automata~\cite{KS15}.
    Finally, for certain quantitative automata over infinite words, \HD nondeterminism adds as much expressiveness as arbitrary nondeterminism~\cite{Col09}.

Recently, pushdown automata on infinite words with \HD nondeterminism ($\omega$-\HD-PDA) were shown to be strictly more expressive than $\omega$-DPDA, while universality and solving games for $\omega$-\HD-PDA are not harder than for $\omega$-DPDA~\cite{LZ22}. 
Thus, \HD nondeterminism adds expressiveness without increasing the complexity of these problems, i.e.\ pushdown automata with \HD nondeterminism induce a novel and intriguing  class of context-free $\omega$-languages. 

Here, we continue this work by studying the expressiveness \emph{and} succinctness of PDA over finite words.
While the decidability results on solving games and universality for $\omega$-\HD-PDA on infinite words also hold for \HD-PDA on finite words, the separation argument between $\omega$-\HD-PDA and $\omega$-DPDA  depends crucially on combining \HD nondeterminism with the coBüchi acceptance condition. Since this condition is only relevant for infinite words, the separation result does not transfer to the setting of finite words.

Nevertheless, we prove that \HD-PDA are more expressive than DPDA, yielding the first class of automata on finite words where \HD nondeterminism adds expressiveness. 
The language witnessing the separation is remarkably simple, in contrast to the relatively subtle argument for the infinitary result~\cite{LZ22}:
the language $\{a^i \$ a^j \$ b^k\$ \mid k \leqslant  \max(i,j)\}$ is recognised by an \HD-PDA but not by a DPDA.
This yields a new class of languages, those recognised by \HD-PDA over finite words, for which universality and solving games are decidable. We also show that this class is incomparable with unambiguous context-free languages.

We then turn our attention to the succinctness of \HD-PDA.
We show that the succinctness gap between DPDA and \HD-PDA is at least exponential, while the gap between \HD-PDA and PDA is at least double-exponential. These results hold already for finite words. 

To the best of our knowledge, both our expressiveness and our succinctness results are the first examples of history-determinism being used effectively over finite, rather than infinite,  words (recall that all \HD-NFA are determinisable by pruning).
Also, this is the first succinctness result for history-deterministic automata that does not depend on the infinitary coBüchi acceptance condition, which was used to show the exponential succinctness of \HD coBüchi automata, as compared to deterministic ones~\cite{KS15}.

We then study an important subclass of \HD-PDA, namely, \HD visibly pushdown automata (VPA), in which the stack behaviour (push, pop, skip) is determined by the input letter only.
\HD-VPA enjoy the good closure properties of VPA (to which they are expressively equivalent): they are closed under complement, union and intersection.
We  show that there is an exponential succinctness gap between deterministic VPA (DVPA) and \HD-VPA, as well as between \HD-VPA and VPA.
Both of these are tight, as VPA, and therefore \HD-VPA as well, admit an exponential determinisation procedure~\cite{AlurM04}.
Furthermore, we show that \HDness of VPA is decidable in $\exptime$.
This makes \HD-VPA a particularly interesting class of PDA as they are recognisable, succinct, have good closure properties and, as for \HD-PDA in general, deciding universality and solving games are both in $\exptime$.
In contrast, solving $\omega$-VPA games is $\twoexptime$-complete~\cite{DBLP:conf/fsttcs/LodingMS04}. We also relate the problem of checking \HDness with the \textit{good-enough synthesis}~\cite{AK20} or \textit{uniformization} problem~\cite{carayol:hal-01806575}, which we show to be \exptime-complete for DVPA and \HD-VPA.

In general, the good compositional properties of \HD automata enable them to be used for solving games and synthesis. In the regular setting, \HD automata are exactly those for which composition with games (of which the winning condition is recognised by the automaton) preserves the winner. Hence, the terms ``good-for-games" and ``history-deterministic" have often been used interchangeably. However, these two notions do not always coincide, for example for quantitative automata~\cite{BL21}. Here we examine the situation for \HD-PDA and show that for PDA, like for ($\omega$-)regular automata, \HD automata are exactly those which compose well with games, that is, history-determinism coincides with good-for-gameness. However, while in the ($\omega$-)regular setting composition with finitely branching and infinitely branching games coincide, this is not the case for PDA: there are PDA that are not \HD but enjoy compositionality with all finitely branching games.

Nondeterminism in \HD automata is resolved on-the-fly, i.e.\ the next transition to be taken only depends on the run prefix constructed so far and the next letter to be processed. 
Thus, the complexity of a resolver, mapping run prefixes and letters to transitions, is a natural complexity measure for \HD automata. 
For example, finite \HD automata (on finite and infinite words) have a finite-state resolver~\cite{HP06}.
For pushdown automata with their infinite configuration space, the situation is markedly different: 
On one hand, we show that \HD-PDA admit positional resolvers, that is, resolvers that depend only on the current configuration, rather than on the entire run prefix produced
so far. Note that this result only holds for \HD-PDA over finite words, but not for $\omega$-\HD-PDA.
Yet, positionality does not imply that resolvers are simple to implement. We show that there are \HD-PDA that do not admit a resolver implementable by a pushdown transducer. In contrast, all \HD-VPA admit pushdown resolvers, again showing that \HD-VPA are better behaved than general \HD-PDA.
Finally, \HD-PDA with finite-state resolvers are determinisable.

\paragraph{Structure of the paper} 
We begin with some preliminaries in Section~\ref{sec_prel} and introduce \HD-PDA in Section~\ref{sec_gfg}.
We study their expressiveness in Section~\ref{sec_expressivity} by comparing them to DPDA, PDA, and unambiguous PDA. 
Afterwards, in Section~\ref{section_succinctness}, we exhibit the succinctness of \HD-PDA in relation to DPDA and PDA. 
The important subclass of \HD-VPA is studied in Section~\ref{section_vpa}, focusing on succinctness, decidability of \HDness, and connections to the good-enough synthesis problem.
Closure properties of \HD-PDA are studied in Section~\ref{sec_closure} while we compare the notions of \HDness and good-for-gameness for PDA on finite words in Section~\ref{sec:composition}.
Then, in Section~\ref{sec_games.tex}, we solve synthesis with specifications given by \HD-PDA as well as universality of such automata.
Next, the resources required to resolve the nondeterminism in an \HD-PDA are studied in Section~\ref{sec_resolvers} before we conclude with some open problems in Section~\ref{section_conc}.

\paragraph{Previous version} A conference version of this paper was presented at MFCS 2021~\cite{GJLZ21}. 
The present article contains all proofs, a more detailed discussion of closure properties, and a new section on compositionality. 

\paragraph*{Related work}
The notion of \HD nondeterminism  has emerged independently several times, at least as Colcombet's history-determinism~\cite{Col09}, in Piterman and Henzinger's \HD automata~\cite{HP06},  and as Kupferman, Safra, and Vardi's nondeterminism for recognising derived languages, that is, the language of trees of which all branches are in a regular language~\cite{KSV06}. Related notions have also emerged in the context of XML document parsing. Indeed, preorder typed visibly pushdown languages and 1-pass preorder typeable tree languages, considered by Kumar, Madhusudan, and Viswanathan~\cite{KMV07} and Martens, Neven, Schwentick, and Bex~\cite{MNSB06} respectively, also consider nondeterminism which can be resolved on-the-fly. However, the restrictions there are stronger than simple \HD nondeterminism, as they also require the typing to be unique, roughly corresponding to unambiguity in automata models and grammars. This motivates the further study of unambiguous \HD automata, although this remains out of scope for the present paper. The XML extension AXML has also inspired Active Context Free Games~\cite{MSS06}, in which one player, aiming to produce a word within a target regular language, chooses positions on a word and the other player chooses a rewriting rule from a context-free grammar. Restricting the strategies of the first player to moving from left to right makes finding the winner decidable~\cite{MSS06,BSSK13}; however, since the player still knows the future of the word, this restriction is not directly comparable to \HD nondeterminism.

Unambiguity, or bounded ambiguity, is an orthogonal way of restricting nondeterminism by limiting the number of permitted accepting runs per word. For regular languages, it leads to polynomial equivalence and containment algorithms~\cite{SH85}. Minimization remains {\sc NP}-complete for both unambiguous automata~\cite{JR93,BM12} and \HD automata~\cite{Sch20} (at least when acceptance is defined on states, see~\cite{AK19}). On pushdown automata, increasing the permitted degree of ambiguity leads to both greater expressiveness and unbounded succinctness~\cite{Her97}.
Finally, let us mention two more ways of measuring--and restricting--nondeterminism in PDA: bounded nondeterminism, as studied by Herzog~\cite{Her97} counts the branching in the run-tree of a word, while the minmax measure~\cite{SY93,GLW05} counts the number of nondeterministic guesses required to accept a word. The natural generalisation of history-determinism as the \textit{width} of an automaton~\cite{MK19} has not yet, to the best of our knowledge, been studied for PDA.

\section{Preliminaries}
\label{sec_prel}
An alphabet~$\Sigma$ is a finite nonempty set of letters. The  empty word is denoted by $\epsilon$, the length of a word~$w$ is denoted by $\size{w}$, and the $n^\text{th}$ letter of $w$ is denoted by $w(n)$ (starting with $n = 0$). The set of (finite) words over $\Sigma$ is denoted by $\Sigma^*$, the set of nonempty (finite) words over $\Sigma$ by $\Sigma^+$, and the set of finite words of length at most $n$ by $\Sigma^{\leq n}$. 
A language over $\Sigma$ is a subset of $\Sigma^*$.

For alphabets~$\Sigma_1,\Sigma_2$, we extend functions~$f \colon \Sigma_1 \rightarrow \Sigma_2^*$ homomorphically to words over $\Sigma_1$ via
$f(w) = f(w(0)) f(w(1)) f(w(2)) \cdots$.

\subsection{Pushdown automata}
\label{sec_definitions}
A pushdown automaton (PDA for short)~$\aut = (Q, \Sigma, \Gamma, q_\initmark, \Delta, F)$
consists of a finite set~$Q$ of states with the initial state~$q_\initmark \in Q$, an input alphabet~$\Sigma$, a stack alphabet~$\Gamma$, a transition relation~$\Delta$ to be specified, and a set~$F$ of final states.
For notational convenience, we define $\Sigma_\epsilon = \Sigma \cup \set{\epsilon}$ and $\Gammabot = \Gamma \cup \set{\bot}$, where $\bot \notin \Gamma$ is a designated stack bottom symbol.
Then, the transition relation~$\Delta$ is a subset of  $Q \times \Gammabot \times \Sigma_\epsilon \times Q \times \Gammabot^{\leqslant2}$ that we require to neither write nor delete the stack bottom symbol from the stack:
if
$(q, \bot, a, q', \gamma) \in \Delta$, then $\gamma \in \bot \cdot (\Gamma \cup \set{\epsilon}) $, and if $(q, X, a, q', \gamma) \in \Delta$ for $X \in \Gamma$, then $\gamma \in \Gamma^{\leqslant2}$. 
Given a transition~$\tau = (q,X,a,q',\gamma)$ let $\ell(\tau) = a \in \Sigma_\epsilon$. 
We say that $\tau$ is an $\ell(\tau)$-transition and that $\tau$ is a $\Sigma$-transition, if $\ell(\tau) \in \Sigma$.
For a finite sequence~$\rho$ over $\Delta$, the word $\ell(\rho) \in \Sigma^*$ is defined by applying $\ell$ homomorphically to every transition.
We take the size of $\aut$ to be $|Q|+|\Gamma|$.\footnote{Note that we prove exponential succinctness gaps, so the exact definition of the size is irrelevant, as long as it is polynomial in $|Q|$ and $|\Gamma|$. Here, we pick the sum for the sake of simplicity.}

A stack content is a finite word in $\bot \Gamma^*$ (i.e. the top of the stack is at the end) and a configuration~$c = (q, \gamma)$ of $\aut$ consists of a state~$q \in Q$ and a stack content~$\gamma$.  
The initial configuration is $(q_\initmark, \bot)$.

The set of modes of $\aut$ is  $Q \times \Gammabot$. A mode~$(q,X)$ enables all transitions of the form~$(q, X, a, q', \gamma' )$ for some $a \in \Sigma_\epsilon$, $q' \in Q$, and $\gamma' \in \Gammabot^{\le 2}$.
The mode of a configuration~$c = (q, \gamma X)$ is $(q,X)$. 
A transition~$\tau$ is enabled by $c$ if it is enabled by $c$'s mode.
In this case, we write~$(q, \gamma X) \trans{\tau} (q', \gamma\gamma')$, where $\tau = (q, X, a, q', \gamma')$.

A run of $\aut$ is a finite sequence~$\rho = c_0 \tau_0 c_1 \tau_1 \cdots c_{n-1} \tau_{n-1} c_n$ of configurations and transitions with $c_0$ being the initial configuration and $c_{n'} \trans{\tau_{n'}}c_{n'+1}$ for every $n' < n$.
The run~$\rho$ is a run of $\aut$ on $w \in \Sigma^*$, if $w = \ell(\rho) $. 
We say that $\rho$ is accepting if it ends in a configuration whose state is final.
The language~$L(\aut)$ recognized by $\aut$ contains all $w \in \Sigma^*$ such that $\aut$ has an accepting run on $w$. 

\begin{rem}
Let $c_0 \tau_0 c_1 \tau_1 \cdots c_{n-1} \tau_{n-1} c_n$ be a run of $\aut$.
Then, the sequence~$c_0c_1 \cdots  c_{n-1} c_n$ of configurations is uniquely determined by the sequence~$\tau_0\tau_1 \cdots \tau_{n-1}$ of transitions.
Hence, whenever convenient, we treat a sequence of transitions as a run if it indeed induces one (not every such sequence does induce a run, e.g. if a transition~$\tau_{n'}$ is not enabled in $c_{n'}$).
\end{rem}

We say that a PDA~$\aut$ is deterministic (DPDA) if 
\begin{itemize}
    \item every mode of $\aut$ enables at most one $a$-transition for every $a \in \Sigma \cup \set{\epsilon}$, and
    
    \item for every mode of $\aut$, if it enables some
    $\epsilon$-transition, then it does not enable any $\Sigma$-transition.
\end{itemize}
Hence, for every input and for every run prefix on it there is at most one enabled transition to continue the run.
Still, due to the existence of $\epsilon$-transitions, a DPDA can have more than one run on a given input. However, these only differ by trailing $\epsilon$-transitions.

The class of languages recognized by PDA is denoted by \cfl, the class of languages recognized by DPDA by \dcfl.

\begin{exa}
\label{example:pda}
The PDA~$\aut$ depicted in Figure~\ref{fig:pdaexample} recognises the language
$\set{a c^nd^n a \mid n \geqslant1} \cup \set{b c^nd^{2n} b \mid n \geqslant1}$.
Note that while $\aut$ is nondeterministic, $L(\aut)$ is in \dcfl.
\begin{figure}
    \centering

\begin{tikzpicture}[thick]
\def\y{1.2}
\def\x{1.2}
\tikzset{every state/.style = {minimum size =22}}
\node[state] (i) at (-.5,0) {};
\node[state] (u) at (2*\x, 0) {$q$};
\node[state] (ad) at (6*\x,\y) {$q_1$};
\node[state] (bd1) at (6*\x,-\y) {$q_2$};
\node[state] (bd2) at (9*\x,-\y) {};
\node[state,fill=lightgray] (acc) at (9*\x,\y) {};

\path[-stealth]
(-.5,-.75) edge (i)
(i) edge node[above] {$a,\bot\mid \bot A$} (u)
(i) edge node[below] {$b,\bot\mid \bot B$} (u)
(u) edge[loop right] node[right] {$c, X\mid XN $} ()
(u.north east) edge[bend left=15] node[near start,above,yshift=.1cm] {$d, N \mid \epsilon$} (ad)
(u.south east) edge[bend right=15] node[near start,below,yshift=-.2cm] {$d, N \mid N$} (bd1)
(ad) edge[loop below] node[] {$d, N \mid \epsilon $} ()
(bd1) edge[bend left=20] node[above] {$d, N \mid \epsilon$} (bd2)
(bd2) edge[bend left=20] node[below] {$d, N \mid N$} (bd1)
(bd2) edge node[right] {$b, B \mid \epsilon$} (acc)
(ad) edge node[above] {$a, A \mid \epsilon$} (acc)
;
\end{tikzpicture}

    \caption{The PDA $\aut$ from Example~\ref{example:pda}. Grey states are final, and $X$ is an arbitrary stack symbol.}
    \label{fig:pdaexample}
\end{figure}

\end{exa}

\section{History-deterministic Pushdown Automata}
\label{sec_gfg}
Here, we introduce history-deterministic push\-down automata on finite words (\HD-PDA for short), nondeterministic pushdown automata whose nondeterminism can be resolved based on the run prefix constructed so far and on the next input letter to be processed, but independently of the continuation of the input beyond the next letter.

As an example, consider the PDA~$\aut$ from Example~\ref{example:pda}. 
It is nondeterministic, but knowing whether the first transition of the run processed an $a$ or a $b$ allows the nondeterminism to be resolved in a configuration of the form~$(q,\gamma N)$ when processing a $d$:
in the former case, take the transition to state~$q_1$, in the latter case the transition to state~$q_2$. 
Afterwards, there are no nondeterministic choices to make and the resulting run is accepting whenever the input is in the language. This automaton is therefore history-deterministic.

We implement this intuition with the notion of a \textit{resolver}, that is, a function that, given the run so far, and the next input letter, produces the sequence of transitions ($\varepsilon$~and otherwise) to be taken to process this letter. After the last letter has been processed, the run produced so far must be accepting, without trailing $\varepsilon$-transitions. We will show that this is not a serious restriction at the end of the section.

Fix a PDA~$\aut = (Q, \Sigma, \Gamma, q_\initmark, \Delta, F)$. 
We say that a run~$c_0 \tau_0 \cdots \tau_n c_{n+1}$ of $\aut$ processing some $w \in \Sigma^*$ has no trailing $\epsilon$-transitions if
\begin{enumerate}
    \item $n =-1$ if $w = \epsilon$, and
    \item $\ell(\tau_0 \cdots \tau_{n-1})$ is a strict prefix of $w$, if $w \neq \epsilon$.
\end{enumerate}

A (nondeterminism) resolver for $\aut$ is a function~$\hstrat \colon \Delta^* \times \Sigma \rightarrow \Delta$ such that 
for every $w \in L(\aut)$, there is an accepting run~$\rho = c_0 \tau_0 \cdots \tau_n c_{n+1}$ on $w$ that has no trailing $\epsilon$-transitions  
and satisfies $\tau_{n'} = \hstrat(\tau_0 \cdots \tau_{n'-1}, w(\size{ \ell(\tau_0 \cdots \tau_{n'-1}) }))$ for all $0 \leqslant n' < n$.
If $w$ is nonempty, then $w(\size{ \ell(\tau_0 \cdots \tau_{n'-1}) })$ is defined for all $0 \leqslant n' < n$, as $\rho$ has no trailing $\epsilon$-transitions.
Note that $\rho$ is unique if it exists.
We say that $\aut$ is history-deterministic (\HD) if it has a resolver.
We denote the class of languages recognised by \HD-PDA by \gfgcfl.

Note that the prefix processed so far can be recovered from $\hstrat$'s input, i.e.\ it is $\ell(\rho)$. However, the converse is not true due to the existence of $\epsilon$-transitions.
This is the reason that the run prefix and not the input prefix is the argument for the resolver.

We require the run induced by a resolver to have no trailing $\epsilon$-transitions, as a resolver requires the next letter to be processed as part of its input. 
This is obviously undefined once the input has ended. 
Nevertheless, at the end of this section, we study a variant of resolvers with end-of-word markers, which will allow trailing $\epsilon$-transitions, showing that they do not add expressiveness.

Intuitively, every DPDA \emph{should} be \HD, as there is no nondeterminism to resolve during a run.
However, in order to reach a final state, a run of a DPDA on some input~$w$ may traverse \emph{trailing} $\epsilon$-transitions after the last letter of $w$ is processed.
On the other hand, the run of an \HD-PDA on $w$ consistent with any resolver has to end with the transition processing the last letter of $w$.
Hence, not every DPDA recognises the same language when viewed as an \HD-PDA. 
Nevertheless, we show, using standard pushdown automata constructions, that every DPDA can be turned into an equivalent \HD-PDA. As every \HD-PDA is a PDA by definition, we obtain a hierarchy of languages. 

\begin{lem}
\label{lemma_inclusions}
\dcfl $\subseteq $ \gfgcfl $\subseteq$ \cfl.
\end{lem}
\begin{proof}
We only consider the first inclusion, as the second one is trivial.
So, let $L \in $ \dcfl, say it is recognised by the DPDA~$\aut = (Q, \Sigma, \Gamma, q_\initmark, \Delta, F)$. 
We say that a mode~$m$ of $\aut$ is a reading mode if it does not enable an $\epsilon$-transition.
Hence, due to determinism, a reading mode $m$ can only enable at most one $a$-transition for every $a \in \Sigma$.

Now, consider some nonempty word~$w(0) \cdots w(n) \in L(\aut)$ (we take care of the empty word later on), say with accepting run~$\rho$ (treated, for notational convenience, as a sequence of transitions).
This run can be decomposed as
\[\rho = \rho_0\, \tau_0\, \rho_1\, \tau_1\, \rho_2\, \cdots\, \rho_n\, \tau_n\, \rho_{n+1} \]
where $\tau_i$ processes $w(i)$ and each $\rho_i$ is a (possibly empty) sequence of $\epsilon$-transitions.
Each run prefix induced by some $\rho_0 \tau_0 \rho_1 \cdots \rho_{i}$ ends in a configuration with reading mode.

Intuitively, we have to eliminate the trailing $\epsilon$-transitions in $\rho_{n+1}$.
To do so, we postpone the processing of each letter~$w(i)$ to the end of $\rho_{i+1}$.
Instead, we guess that the next input is $w(i)$ by turning the original  $w(i)$-transition~$\tau_i$ into an $\epsilon$-transition $\tau_i'$ that stores $w(i)$ in the state space of the modified automaton. 
Then, $\rho_{i+1}$ is simulated and a dummy transition~$\tau_i^d$ processing the stored letter~$w(i)$ is executed. 

Hence, the resulting run of the modified automaton on $w$ has the form
\[\rho_0\, \tau_0'\, \rho_1\, \tau_0^d\, \tau_1'\, \rho_2\, \tau_1^d\, \cdots\, \rho_n\, \tau_{n-1}^d\, \tau_n'\, \rho_{n+1}\, \tau_n^d,\]
where each $\tau_i'$ is now an $\epsilon$-transition, each $\rho_i$ is a (possibly empty) sequence of $\epsilon$-transitions, and each $\tau_i^d$ is a dummy transition processing $w(i)$. 
Hence, the run ends with the transition processing the last letter $w(n)$ of the input.

The resulting PDA is \HD, as a resolver has access to the next letter to be processed, which is sufficient to resolve the nondeterminism introduced by the guessing of the next letter.

More formally, consider the PDA~$\aut' = (Q', \Sigma, \Gamma, q_\initmark', \Delta', F')$ where 
\begin{itemize}
    \item $Q' = Q \cup (Q \times \Sigma)$,
    \item $q_\initmark' = q_\initmark$,
    \item $F' = F \cup I$ where $I=\set{q_\initmark}$ if $\epsilon \in L(\aut)$ and $I = \emptyset$, otherwise, and
    \item $\Delta'$ is the union of the following sets of transitions:
    \begin{itemize}
        \item $\set{\tau \in \Delta \mid \ell(\tau) = \epsilon}$, which is used to simulate the leading sequence of $\epsilon$-transitions before the first letter is processed by $\aut$, i.e. the transitions in $\rho_0$ above.

        \item $\set{(q, X, \epsilon, (q',a),\gamma) \mid (q,X,a,q',\gamma) \in \Delta \text{ and } a\in \Sigma}$, which are used to guess and store the next letter to be processed, i.e. the transitions $\tau_i'$ above.

        \item $\set{((q,a), X, \epsilon, (q',a),\gamma) \mid (q,X,\epsilon,q',\gamma) \in \Delta }$, which are used to simulate $\epsilon$-transitions after a letter has been guessed, but not yet processed, i.e. transitions in some $\rho_i$ with $i>0$ above.
        
        \item $\set{ ((q,a), X, a, q, X) \mid (q,X) \text{ is a reading mode}}$, the dummy transitions used to actually process the guessed and stored letter.
        
    \end{itemize}
    
\end{itemize}

Now, formalising the intuition given above, one can show that $\aut'$ has a resolver witnessing that it recognises $L(\aut)$.
In particular, the empty word is in $L(\aut')$ if and only if it is in $L(\aut)$, as the run induced by the resolver on $\epsilon$ ends in the initial configuration, which is final if and only if $\epsilon \in L(\aut)$.
\end{proof}

\HD-PDA are by definition required to end their run with the last letter of the input word. 
Instead, one could also consider a model where they are allowed to take some trailing $\epsilon$-transitions after the last input letter has been processed.
As a resolver has access to the next input letter, which is undefined in this case, we need resolvers with end-of-word markers to signal the resolver that the last letter has been processed. 
In the following, we show that \HD-PDA with end-of-word resolvers are as expressive as standard \HD-PDA, albeit exponentially more succinct.

Fix some distinguished end-of-word-marker~$\#$, which takes the role of the next input letter to be processed, if there is none after the last letter of the input word is processed. 
Let $\aut = (Q, \Sigma, \Gamma, q_\initmark, \Delta, F)$ be a PDA with $\# \notin \Sigma$.
An EoW-resolver for $\aut$ is a function~$\hstrat \colon \Delta^* \times (\Sigma \cup \set{\#}) \rightarrow \Delta$ such that for every $w \in L(\aut)$, there is an accepting run~$c_0 \tau_0 \cdots \tau_n c_n$ on $w$ such that $\tau_{n'} = \hstrat(\tau_0 \cdots \tau_{n'-1}, w\#(\size{ \ell(\tau_0 \cdots \tau_{n'-1}) }))$ for all $0 \leqslant n' < n$.
Note that the second argument given to the resolver is a letter of $w\#$, which is equal to $\#$ if the run prefix induced by $\tau_0 \cdots \tau_{n'-1}$ has already processed the full input~$w$.

\begin{lem}
\label{lemma_eow}
\HD-PDA with EoW-resolvers are as expressive as \HD-PDA.
\end{lem}

\begin{proof}
A (standard) resolver can be turned into an EoW-resolver that ignores the EoW-marker. 
Hence, every \HD-PDA is an \HD-PDA with EoW-resolver recognizing the same language.
So, it only remains to consider the other inclusion.

To this end, let $\aut = (Q, \Sigma, \Gamma, q_\initmark, \Delta, F)$ be a PDA with EoW-resolver.
The language
\[
C_{acc} = \set{ \gamma q \mid q \in F \text{ and } \gamma \in \bot\Gamma^* } \subseteq  \bot\Gamma^*Q
\]
encoding final configurations of $\aut$ is regular. 
Hence, the language
\begin{align*}
C = \{ \gamma q \in \bot\Gamma^*Q \mid \text{ there is a run infix } (q, \gamma) \tau_0 \cdots \tau_{n-1} c_n \ \\ 
\text{ with } \ell(\tau_0 \cdots \tau_{n-1}) = \epsilon \text{ and } c_n \in C_{acc} \}
\end{align*}
can be shown to be regular as well by applying saturation techniques~\cite{Buechi1964}\footnote{Also, see the survey by Carayol and Hague~\cite{CarayolHague14} for more details.} to the restriction of $\aut$ to $\epsilon$-transitions.
If $\aut$ reaches a configuration~$c \in C$ after processing an input~$w$, then $w \in L$, even if $c$'s state is not final.

Let $\auta = (Q_\auta, \Gammabot \cup Q, q_\initmark^\auta, \delta_\auta, F_\auta)$ be a DFA recognizing $C$.
We extend the stack alphabet of $\aut$ to $\Gamma \times  Q_\auta \times (Q_\auta \cup \set{u})$, where $u$ is a fresh symbol.
Then, we extend the transition relation such that it keeps track of the unique run of $\auta$ on the stack content: If $\aut$ reaches a stack content~$\bot(X_1, q_1, q_1') (X_2, q_2, q_2') \cdots (X_s, q_s, q_s') $, then we have 
\[
q_j = \delta_\auta^*(q_\initmark
^\auta,\bot X_1 \cdots X_j)
\]
for every $1 \leqslant j \leqslant s$ as well as $q_j' = q_{j-1}$ for every $2 \leqslant j \leqslant s$ and $q_1' = u$. Here, $\delta_\auta^*$ is the standard extension of $\delta_\auta$ to words. 
The adapted PDA is still \HD, as no new nondeterminism has been introduced, and keeps track of the state of $\auta$ reached by processing the stack content as well as the shifted sequence of states of $\auta$, which is useful when popping the top stack symbol:
If the topmost stack symbol~$(X, q, q')$ is popped of the stack then $q'$ is the state of $\auta$ reached when processing the remaining stack.

Now, we double the state space of $\aut$, making one copy final, and adapt the transition relation again so that a final state is reached whenever $\aut$ would reach a configuration in $C$.
Whether a configuration in $C$ is reached can be determined from the current state of $\aut$ being simulated, as well as the top stack symbol containing information on the run of $\auta$ on the current stack content.
The resulting PDA~$\aut'$ recognises $L(\aut)$ and has on every word~$w \in L(\aut)$ an accepting run without trailing $\epsilon$-transitions. 
Furthermore, an EoW-resolver for $\aut$ can be turned into a (standard) resolver for $\aut'$, as the tracking of stack contents and the doubling of the state space does not introduce nondeterminism.
\end{proof}

As $\auta$ has at most exponential size,  $\aut'$ is also exponential (both in the size of $\aut$).
This exponential blowup incurred by removing the end-of-word marker is in general unavoidable.
In Theorem~\ref{theorem:succinctness}, we show that the language~$L_n$ of bit strings whose $n^\text{th}$ bit from the end is a $1$ requires exponentially-sized \HD-PDA.
On the other hand, it is straightforward to devise a polynomially-sized \HD-PDA~$\aut_{\text{EoW}}$ with EoW-marker recognizing $L_n$: the underlying PDA stores the input word on the stack, guesses nondeterministically that the word has ended, uses $n$ (trailing) $\epsilon$-transitions to pop of the last $n-1$ letters stored on the stack, and then checks that the topmost stack symbol is a $1$. 
With an EoW-resolver, the end of the input does not have to be guessed, but is marked by the EoW-marker.
Hence, $\aut_{\text{EoW}}$ is \HD.

Finally, let us remark that the history-determinism of PDA and context-free languages is undecidable.
These problems were shown to be undecidable for $\omega$-\HD-PDA and $\omega$-\HD-CFL by reductions from the inclusion and universality problem for PDA on finite words~\cite[Theorem 6.1]{LZ22}. The same reductions also show that these problems are undecidable over PDA on finite words.

\begin{thm}
The following problems are undecidable:
\begin{enumerate}
    \item Given a PDA~$\aut$, is $\aut$ an \HD-PDA?
    \item Given a PDA~$\aut$, is $L(\aut) \in $ \gfgcfl? 
\end{enumerate}
\end{thm}

\section{Expressiveness}
\label{sec_expressivity}
Here we show that \HD-PDA are more expressive than DPDA but less expressive than PDA.

\begin{thm} \label{thm:expressiveness}
\dcfl $\subsetneq$ \gfgcfl $\subsetneq$ \cfl.
\end{thm}

To show that \HD-PDA are more expressive than deterministic ones, we consider the language~$
    B_2 = \{a^i \$ a^j \$ b^k \$ \mid k \leqslant \max(i,j)\}$. It is recognised by the PDA $\aut_{B_2}$
depicted in Figure \ref{fig:gfgpdaexample}, hence $B_2 \in $ \cfl. 
We show that $B_2 \in $ \gfgcfl by proving that $\aut_{B_2}$ is \hd: the only nondeterministic choice, between moving to $p_1$ or to $p_2$ upon reading the second $\$$, can be made only based on the prefix~$a^i\$ a^j$ processed so far, which deterministically lead to $q_2$.

\begin{figure}
    \centering
\begin{tikzpicture}[thick]
\def\y{.9}
\def\x{1.5}
\def\b{10}
\tikzset{every state/.style = {minimum size =22}}
\node[state] (i) at (0*\x,0) {$q_1$};
\node[state] (u) at (2*\x, 0) {$q_2$};
\node[state] (ad) at (4*\x,0) {$p_1$};
\node[state] (bd) at (6*\x,0) {$p_2$};
\node[state,fill=lightgray] (fin) at (8*\x,0) {$f$};

\path[-stealth]
(-.75,0) edge (i)
(i) edge[loop above] node[] {$a, X\mid Xa$} ()
 (i) edge node[below] {$\$,X\mid X\$$} (u)
 (bd) edge[] node[below] {$\$,X\mid X\$$} (fin)
 (u) edge[loop above] node[] {$a, X\mid Xa$} ()
 (u) edge[] node[below,near end,xshift=-.2cm] {$\$, X \mid X$} (ad)
 (u) edge[bend right] node[below] {$\$, X \mid X$} (bd)
 (ad) edge[loop above] node[] {$\epsilon, a \mid \epsilon $} ()
 (bd) edge[loop above] node[] {$b, a \mid \epsilon $} ()
 (ad) edge[] node[below,near start] {$\epsilon, \$ \mid \epsilon$} (bd)
;
\end{tikzpicture}

    \caption{A PDA $\aut_{B_2}$ recognising $B_2$.
    Grey states are final, and $X$ is an arbitrary stack symbol.}
    \label{fig:gfgpdaexample}
\end{figure}

\begin{lem} \label{lem_expressiveness}
    $B_2 \in$ \gfgcfl.
\end{lem}

\begin{proof}
Let us summarise the behaviour of the pushdown automaton $\aut_{B_2}$
recognising $B_2$ (see Figure \ref{fig:gfgpdaexample}).
First, the automaton copies the two blocks of $a$'s on
the stack.
Then, when it processes the second $\$$,
it transitions nondeterministically to either $p_1$ or $p_2$.
In $p_1$, it erases the second $a$-block from the stack,
so that the first block is at the top of the stack,
and then transitions to $p_2$.
In $p_2$, the automaton compares the number of $b$'s in the input
with the number of $a$'s in the topmost block of the stack.
If the latter is larger than or equal to the former,
$\aut_{B_2}$ pops one $a$ for each $b$ in the input,
and then transitions to the final state when it processes
the third $\$$.

When processing the second $\$$,
knowing whether the first or second block of the prefix contains more $a$'s
allows the nondeterminism to be resolved:
if the first block contains more $a$'s, take the transition to the state $p_1$,
if the second block contains more $a$'s, take the transition to the state $p_2$.
\end{proof}

Now, in order to show that $B_2$ is not in $\dcfl$,
we prove that its complement $B_2^c$ is not a context-free language.
Since $\dcfl$ is closed under complementation, this implies the desired result.

\begin{lem} \label{lem_not_DCFL}
    The complement $B_2^c$ of $B_2$ is not in $\cfl$.
\end{lem}

\begin{proof}
    Assume, for the sake of contradiction, that
    the complement $B_2^{c}$ of $B_2$ is in $\cfl$.
    Now consider the regular language
    \begin{align*}
    A &= \{a^i \$ a^j \$ b^k \$| i,j,k \in \mathbb{N}\}.
    \end{align*}
    Since the intersection of a context-free language and a regular language is context-free,
    we have that $B_2^{c}\cap A \in \cfl$.
    Therefore, $B_2^{c} \cap A$ satisfies the pumping lemma for context-free languages~(see, e.g.~\cite{Hopcroft}):
    there exists $m \in \mathbb{N}$ such that the word $z = a^m \$ a^m \$ b^{m+1}\$ \in B_2^{c} \cap A$
    can be decomposed as $z = uvwxy$ such that
    \begin{enumerate}
        \item\label{item:pump_empty} $|vx| \geqslant 1$;
        \item\label{item:pump_iter} $uv^nwx^ny \in B_2^{c}\cap A$ for every $n \geqslant 0$.
    \end{enumerate}
    Note that Item \ref{item:pump_iter} directly implies that both $v$ and $x$ are in the language
    $\{a\}^* \cup \{b\}^*$,
    as otherwise $uv^2wx^2y$ is not in $A$.
    On top of that, Item~\ref{item:pump_empty} implies that either $v$ or $x$ is in
    $\{a\}^+ \cup \{b\}^+ $.
    We conclude by proving, through a case distinction, that Item~\ref{item:pump_iter}
    cannot hold as either $uwy$ or $uv^2wx^2y$ is in $B_2$.
    \begin{itemize}
        \item
        Assume that neither $v$ nor $x$ is in $\{b\}^+$.
        Then either $v$ or $x$ is in $\{a\}^+$,
        hence $uv^2wx^2y= a^{m_1} \$ a^{m_2} \$ b^{m+1}\$$ for some $m_1,m_2 \geqslant m$
        such that either $m_1>m$ or $m_2>m$.
        In both cases, we get $uv^2wx^2y \in B_2$.
        \item
        Assume that either $v$ or $x$ is in $\{b\}^+$.
        Then pumping $v$ and $x$ down in $z$ reduces the size
        of at most one of the $a$-blocks:
        we have that $uwy = a^{m_1} \$ a^{m_2} \$ b^{m_3 +1}\$$
        for $m_1,m_2,m_3 \leqslant m$ such that $m_3<m$, and either $m_1 = m$ or $m_2 = m$.
        In both cases, we get $uwy \in B_2$.
    \end{itemize}
    As every possible case results in a contradiction, 
    our initial hypothesis is false:
    $B_2^c \not\in$ \cfl.
\end{proof}

The previous two lemmata and Lemma~\ref{lemma_inclusions} yield \dcfl $\subsetneq$ \gfgcfl.

Finally, to show that PDA are more expressive than \HD-PDA, we consider the language $L = \{a^nb^n \mid n \geqslant 0\} \cup \{a^nb^{2n} \mid n \geqslant 0\}$. 
We note that $L \in $ \cfl while we show below $L \notin $ \gfgcfl using arguments similar to the classical proof showing that $L$ is not \dcfl.
Hence, the following lemma completes the proof of Theorem~\ref{thm:expressiveness}.

\begin{lem} \label{lem:expressiveness1}
$L \notin $ \gfgcfl.
\end{lem}

\begin{proof}
We show that there does not exist an \HD-PDA recognising $L$.
In fact, we show that if there exists an \HD-PDA $\aut$ recognising $L$, then we can construct a PDA $\widehat{\aut}$ recognising the language $\widehat{L}=L \cup \{a^n b^n c^n  \mid  n \geqslant 0\}$.
A straightforward application of the pumping lemma for context-free languages~(see, e.g.~\cite{Hopcroft}) to words of the form~$a^nb^nc^n$ shows that $\widehat{L}$ is not in \cfl. 
Thus, we reach a contradiction.

The idea behind the construction is to replicate the part of the control unit of $\aut$ which processes the suffix $b^n$ of an input word $a^nb^{2n}$ with the difference that in the newly added parts, the transitions caused by input symbol $b$ are replaced with similar ones for input symbol $c$.
This new part of the control unit may be entered after $\widehat{\aut}$ has processed $a^nb^n$.

We now construct the PDA $\widehat{\aut}$ from $\aut$ as follows.
Let $\aut = (Q, \Sigma, \Gamma, q_I, \Delta, F)$ with $Q=\{q_0, q_1, \ldots, q_n\}$, and let $q_0=q_\initmark$.
Now consider $\widehat{\aut} = (Q \cup \widehat{Q}, \Sigma, \Gamma, q_I, \Delta \cup \widehat{\Delta}, F \cup \widehat{F})$ with
$\widehat{Q} =\{\widehat{q}_0, \widehat{q}_1, \ldots, \widehat{q}_n\}$, $\widehat{F} =\{\widehat{q_i}  \mid  q_i \in F\}$,
and $\widehat{\Delta}$ includes the following additional transitions:
\begin{enumerate}[1.]
    \item $\{(q_f, X, \epsilon, \widehat{q}_f, X) \mid q_f \in F, X \in \Gammabot\}$: switch from the original final states to the new states.
    \item $\{(\widehat{q}_i, X, c, \widehat{q}_j, \gamma) \mid (q_i, X, b, q_j, \gamma) \in \Delta \}$: replicate the original $b$-transitions by $c$-transitions in the new states.
    \item \{$(\widehat{q}_i, X, \epsilon, \widehat{q}_j, \gamma) \mid (q_i, X, \epsilon, q_j, \gamma) \in \Delta \}$: replicate all $\epsilon$-transitions. 
\end{enumerate}

Now we show that $L(\widehat{\aut})=\widehat{L}$.
First we show that $L(\widehat{\aut}) \subseteq \widehat{L}$.
Consider a word $w \in L(\widehat{\aut})$.
There may be two cases:
\begin{enumerate}[(i)]
    \item Assume $\widehat{\aut}$ has an accepting run on $w$ that does not visit a state in $\widehat{Q}$. In this case, we have that $w$ is in $L(\aut)=L\subseteq \widehat{L}$.
    
    \item Assume there exists an accepting run of $\widehat{\aut}$ on $w$ that visits a state in $\widehat{Q}$.
    Since $\aut$ recognises $L$, and by construction of $\widehat{\aut}$, a state $\widehat{q}_i \in \widehat{Q}$ can be reached from a state $q_i \in Q$ only if $\widehat{q}_i \in \widehat{F}$ and $q_i \in F$, and the corresponding transition is an $\epsilon$-transition, we have that starting from the initial configuration $(q_\initmark, \bot)$, a state in $\widehat{Q}$ is reached for the first time only after processing an input prefix $a^nb^n$ or $a^nb^{2n}$ for some $n \geqslant 0$.
    If this prefix of $w$ is $a^nb^{2n}$, then $w=a^nb^{2n}$.
    This is because if $w=a^nb^{2n}c^m$ for some $m > 0$ (recall that after visiting a state $\widehat{q}_i$ in $\widehat{Q}$, the only non-$\epsilon$ transitions possible are on the letter $c$), then by the construction of $\widehat{\aut}$, we have that $\aut$ can accept the word $a^nb^{2n}b^m$ which is not in the language $L$.
    On the other hand, let the prefix be $a^nb^n$ when a state $\widehat{q}_i \in \widehat{Q}$ is visited for the first time.
    Note that $\widehat{q}_i \in \widehat{F}$, and  let $(\widehat{q}_i, \gamma_i)$ be the corresponding configuration.
    If a sequence of transitions $\widehat{\tau}_i, \ldots, \widehat{\tau}_j$ from $(\widehat{q}_i, \gamma_i)$ to $(\widehat{q}_j, \gamma_j)$ is possible such that not all of $\widehat{\tau_i}, \ldots , \widehat{\tau_j}$ are $\epsilon$-transitions, that is, the transitions process $c^m$ for some $m \in \nats$, and $\widehat{q}_j \in \widehat{F}$, then a sequence of transitions $\tau_i, \ldots, \tau_j$ of the same length processing $b^m$ is possible from $(q_i, \gamma_i)$ to $(q_j, \gamma_j)$ with $q_j \in F$.
    Since this leads to an accepting run from $(q_\initmark, \bot)$ to $(q_j, \gamma_j)$ while visiting only the states in $Q$ on processing $a^nb^nb^m$ with $m > 0$, we have $m=n$, and hence $w=a^nb^nc^n \in \widehat{L}$.
    
    On the other hand, if all transitions $\widehat{\tau_i}, \ldots, \widehat{\tau_j}$ are $\epsilon$-transitions, then $w=a^nb^n \in \widehat{L}$.
\end{enumerate}

Now we prove the other direction, that is $\widehat{L} \subseteq L(\widehat{\aut})$.
Here, we rely on the fact that the accepting run of $\aut$ on $a^nb^n$ induced by a resolver~$\hstrat$ is a prefix of the accepting run of $\aut$ on $a^nb^{2n}$ induced by $\hstrat$. This allows to switch to the copied states~$\widehat{Q}$ after processing $a^nb^n$ and then process $c^n$ instead of $b^n$.

Consider a word $w \in \widehat{L}$ such that $w \in L$.
By construction of $\widehat{\aut}$, we have that $w \in L(\widehat{\aut})$ since $\widehat{\aut}$ accepts all words that are also accepted by $\aut$.
Now suppose that $w \in \widehat{L}$ but $w \notin L$, that is, $w$ is of the form $a^nb^nc^n$ for some $n \geqslant 1$.
Since by assumption, we have that $\aut$ is an \HD-PDA recognising the language $L$, there exists a resolver $r$ that for every word in $L$ induces an accepting run of the word in $L$.
Let  $(q_i, \gamma_i)$ be  the configuration of $\aut$ reached after processing the prefix $a^nb^n$ in the run induced by $r$ on the input $a^nb^{2n}$. 

Note that $q_i \in F$ since $r$ also induces an accepting run for the input $a^nb^n$.
Now if for the input $a^nb^{2n}$, the sequence of transitions chosen by $r$ from $(q_i, \gamma_i)$ after processing $a^nb^n$ is $\tau_i, \tau_{i+1}, \ldots, \tau_j$ with $(q_i, \gamma_i) \trans{\tau_{i}} (q_{i+1}, \gamma_{i+1}) \cdots (q_{j-1}, \gamma_{j-1}) \trans{\tau_j} (q_j,\gamma_j)$, with $q_j \in F$, and the sequence $\tau_i, \ldots, \tau_j$ processes $b^n$, then by the construction of $\widehat{\aut}$, there exists a sequence of transitions $\widehat{\tau}_{i}, \widehat{\tau}_{i+1}, \ldots, \widehat{\tau}_j$ with $(\widehat{q}_i, \gamma_i) \trans{\widehat{\tau}_i} (\widehat{q}_{i+1}, \gamma_{i+1}) \cdots (\widehat{q}_{j-1}, \gamma_{j-1}) \trans{\widehat{\tau}_j} (\widehat{q}_j, \gamma_j)$ and with $\widehat{q}_j \in \widehat{F}$ such that there is an $\epsilon$-transition from $(q_i, \gamma_i)$ to $(\widehat{q}_{i}, \gamma_{i})$ and the sequence $\widehat{\tau}_i, \widehat{\tau}_{i+1}, \ldots, \widehat{\tau}_j$ processes $c^n$, and hence $w \in L(\widehat{\aut})$.

Thus we have that $\widehat{L}=L(\widehat{\aut})$.
Hence we show that if $\aut$ is an \HD-PDA, then we can construct a PDA $\widehat{\aut}$ recognising $\widehat{L}$ which is not a CFL, thus leading to a contradiction to our assumption that $L$ is in \HD-CFL.
\end{proof}

Unambiguous context-free languages, i.e.\ those generated by grammars for which every word in the language has a unique leftmost derivation, are another class sitting between \dcfl and \cfl. 
Thus, it is natural to ask how unambiguity and history-determinism are related: To conclude this section, we show that both notions are independent.

\begin{thm}
\label{unamb_vs_gfg}
There is an unambiguous context-free language that is not in \gfgcfl and a language in \gfgcfl that is inherently ambiguous. 
\end{thm}

An unambiguous grammar for the language~$\{a^nb^n \mid n \geqslant 0\} \cup \{a^nb^{2n} \mid n \geqslant 0\} \notin$ \gfgcfl is easy to construct and we show that the language~$B=\{a^i b^j c^k \mid i,j,k \geqslant 1, k \leqslant \max(i,j)\}$ is inherently ambiguous.
Its inclusion in \gfgcfl is easily established using a similar argument as for the language~$B_2 = \{a^i \$ a^j \$ b^k \$ \mid k \leqslant \max(i,j)\}$ above. The dollars add clarity to the \HD-PDA but are cumbersome in the proof of inherent ambiguity.

We show that $B=\{a^i b^j c^k \mid i,j,k \geqslant 1, k \leqslant \max(i,j)\}$ is inherently ambiguous, i.e.\ for every grammar generating $B$ there is at least one word that has two different leftmost derivations.

We use standard definitions and notation for context-free grammars as in~\cite[Section~4.2]{Hopcroft}. We say that a grammar is reduced, if every variable is reachable from the start variable, every variable can be reduced to a word of terminals, and for no variable $A$, it holds that $A {\overset{*}{\implies}}A$.

Let $D(G) = \{A \in V \mid A  {\overset{*}{\implies}} x A y \text{ for some } x,y \text{ with } xy\neq \epsilon\}$.
An unambiguous CFG $G$ is called \emph{almost-looping}, if
\begin{enumerate}
    \item $G$ is reduced,
    \item all variables, possibly other than the start variable $S$, belong to $D(G)$, and
    \item either $S \in D(G)$ or $S$ occurs only once in the leftmost derivation of any word in $L(G)$.
\end{enumerate}

Now we state the following lemma from \cite{Maurer69}.
\begin{lem} \label{lem:Maurer}
For every unambiguous CFG $G$, there exists an unambiguous almost-looping CFG $G'$ such that $L(G)=L(G')$.
\end{lem}


An example of an almost-looping grammar for language $B$ is the following:
\begin{figure}[h]
\hfill
\begin{minipage}{0.4\textwidth}
$S \rightarrow S_1  \mid  S_2$ \\
$S_1 \rightarrow aS_1  \mid  aS_1c  \mid  aBc$\\
$B \rightarrow bB  \mid  b$
\end{minipage}
\hfill
\begin{minipage}{0.4\textwidth}
$S_2 \rightarrow aS_2  \mid  aD$\\
$D \rightarrow bDc  \mid  bD  \mid  bc$
\end{minipage}
\hfill
\caption{An example CFG for language $B$}
\label{fig:cfg1}
\end{figure}

Now we prove the following, using techniques inspired by Maurer's proof that $\set{a^i b^j c^ k \mid i,j,k \geqslant 1, i=j \vee j=k}$ is inherently ambiguous~\cite{Maurer69}.

\begin{lem}
The language $B$ is inherently ambiguous.
\end{lem}

\begin{proof}
Assume, towards a contradiction, that $G$ is an unambiguous grammar for $B$, which, from Lemma~\ref{lem:Maurer}, we can assume, without loss of generality, to be an almost-looping grammar.
Let $A$ be a variable of $G$.
\begin{enumerate}
    
    \item \label{onlya} $A$ is of Type~\ref{onlya} if there is a derivation~$A {\overset{*}{\implies}} xAy$ where $xy=a^{\markingn{A}{1}}$ for some $\markingn{A}{1} > 0$.

    \item \label{onlyb} $A$ is of Type~\ref{onlyb} if there is a derivation~$A {\overset{*}{\implies}} xAy$ where $xy=b^{\markingn{A}{2}}$ for some $\markingn{A}{2} > 0$.
    
    \item \label{ac} $A$ is of Type~\ref{ac} if there is a derivation~$A {\overset{*}{\implies}} xAy$ where $x=a^{\markingl{A}{3}}$ and $y=c^{\markingr{A}{3}}$ for some $\markingl{A}{3} \geqslant \markingr{A}{3} > 0$.
    
    \item \label{bc} $A$ is of Type~\ref{bc} if there is a derivation~$A {\overset{*}{\implies}} xAy$ where $x=b^{\markingl{A}{4}}$ and $y=c^{\markingr{A}{4}}$ for some $\markingl{A}{4} \geqslant \markingr{A}{4} > 0$.
    
    \item \label{ab} $A$ is of Type~\ref{ab} if there is a derivation~$A {\overset{*}{\implies}} xAy$ where $x=a^{\markingl{A}{5}}$ and $y=b^{\markingr{A}{5}}$ for some $\markingl{A}{5}, \markingr{A}{5} > 0$.    
\end{enumerate}
Note that some variables may be of multiple types (e.g.\ the variable~$D$ in Figure~\ref{fig:cfg1}
 has Type~\ref{onlyb} and Type~\ref{bc}).
 
First, we show that each variable in $D(G)$ has at least one of these five types.
So, let $A \in D(G)$. Then, there exists a derivation~$A {\overset{*}{\implies}} xAy$ with $xy\neq\epsilon$.
Note that both $x$ and $y$ belong to $a^*$, $b^*$, or $c^*$ since otherwise, due to $G$ being reduced, one could derive words that are not in the language.
Next, we note that the cases where $x$ belongs to $c^*$, and $y$ belongs to $a^*$ or $b^*$ cannot happen.
Similarly, the case where $x$ belongs to $b^*$, and $y$ belongs to $a^*$ cannot happen.
Also we cannot have $xy$ in $c^*$, since this will allow us to have words with an arbitrary number of $c$'s which can be more than the number of $a$'s and $b$'s and such a word is not in the language.

Further, we cannot have $x=a^{\ell}$ and $y=c^{r}$ with $0 < \ell < r$. Otherwise, consider a derivation of some word in $B$ that uses $A$, i.e.\ \[S {\overset{*}{\implies}} \alpha A \beta {\overset{*}{\implies}} a^s b^u c^v \text{\quad with $v \leqslant \max(s,u)$.}\]  Now, towards a contradiction assume we indeed have \[A {\overset{*}{\implies}} x A y \text{\quad with $x=a^{\ell}, y=c^{r}$ and $\ell < r$.}\] 
Then, pumping $q$ copies of $x$ and $y$, for some suitable $q \in\nats$, yields a derivation
\[
S {\overset{*}{\implies}} \alpha A \beta {\overset{*}{\implies}} \alpha x^q A y^q \beta {\overset{*}{\implies}} a^{s+\ell q} b^{u} c^{v+r q}
\text{\quad such that $v+rq > \max(s+\ell q, u)$},
\]
i.e.\ we have derived a word that is not in $B$.
Similarly, we cannot have $x=b^\ell$ and $y=c^r$ for some $0 < \ell < r$.
Altogether, this implies that $A$ indeed has at least one of the five types stated above.

Moreover, we claim that there is a $t \in \nats$ such that the following three properties are true for every word~$w \in B$:
\begin{description}
    \item[Property 1] If $w$ has more than $t$ $c$'s, then the (unique) leftmost derivation of $w$ has the form  
    \[S {\overset{*}{\implies}} \alpha A \beta {\overset{*}{\implies}} \alpha x A y \beta {\overset{*}{\implies}} w
    \text{\quad such that $xy$ contains a $c$. }\]
    Thus, $A$ has type~\ref{ac} or type~\ref{bc}.

\item [Property 2] If $w$ has more than $t$ $a$'s, then the (unique) leftmost derivation of $w$ has the form  
    \[S {\overset{*}{\implies}} \alpha A \beta {\overset{*}{\implies}} \alpha x A y \beta {\overset{*}{\implies}} w
    \text{\quad such that $xy$ contains an $a$. }\]
    Thus, $A$ has type~\ref{onlya}, type~\ref{ac}, or type~\ref{ab}.
    
\item [Property 3] If $w$ has more than $t$ $b$'s, then the (unique) leftmost derivation of $w$ has the form  
    \[S {\overset{*}{\implies}} \alpha A \beta {\overset{*}{\implies}} \alpha x A y \beta {\overset{*}{\implies}} w
    \text{\quad such that $xy$ contains a $b$. }\]
    Thus, $A$ has type~\ref{onlyb}, type~\ref{bc}, or type~\ref{ab}.
\end{description}

We prove these properties as follows:
we denote by $d$ the \emph{width} of the grammar $G$
which is the maximum number of symbols appearing on the right side of some production rule of $G$.
Further, we denote by $m$ the number of variables appearing in $G$.
We argue that $t = d^{m+1}$ satisfies the three properties above.
We focus on Property~1, the two other proofs are similar.
Suppose that $w$ contains more than $d^{m+1}$ $c$'s and consider the derivation tree of that word.
The \emph{weight}~$\omega(v)$ of a vertex~$v$ in the derivation tree is defined as the number of $c$'s in the subtree rooted at~$v$.
Hence, the root of the derivation tree has at least weight~$d^{m+1}$.
We build a finite path $v_0,v_1, \ldots, v_k$
from the root of this tree to one of its leaves as follows:
The initial vertex~$v_0$ is the root and
at each step, we choose as successor of $v_{i}$
its child $v_{i+1}$ with the largest weight.
A vertex~$v_{i}$ of this path is \emph{decreasing} if $\omega(v_{i}) > \omega(v_{i+1})$.
There are are at least $m+1$ decreasing vertices on the path
because 
$\omega(v_0) = d^{m+1}$,
$\omega(v_k) = 1$,
and $\omega(v_{i+1}) \geq \frac{1}{d} \cdot \omega(v_i)$.
Thus, there are two decreasing vertices on the path that are labelled by the same variable~$A$ such that there is a derivation of the form~$A {\overset{*}{\implies}} xAy$ with some $c$ in $xy$.

Let $p > t$ be a positive integer divisible by the least common multiple of the $\markingn{A}{i}, \markingl{A}{i}$ and $\markingr{A}{i}$ for all $A \in D(G)$ and $i \in \set{1, \ldots, 5}$, where we define $\markingn{A}{1} = 1$ if $A$ is not of Type~\ref{onlya}, and similarly for all other~$i > 1$.
We show that the word $w = a^{2p}b^{2p}c^{2p} \in B$ has two leftmost derivations.

First consider the derivation of the word $w_b = a^{2p}b^{p}c^{2p} \in B$.
As we have more than $t$ $c$'s in $w_b$ Property~1 shows that the derivation contains a variable of Type~\ref{ac} or Type~\ref{bc}. 
Next, we argue that it cannot contain a variable of Type~\ref{bc}: The occurrence of such a variable would allow us to either produce a word that is not in $a^*b^*c^*$ or to inject $b^{p}c^{r}$ for $p \geqslant r > 0$ leading to the derivation of $a^{2p}b^{2p}c^{2p+r}$, which is not in the language.
Thus, the derivation of $w_b$  uses at least one variable of Type~\ref{ac}.
Also, since $w_b$ has $p > t$ $b$'s,
Property~3 implies that the (unique) leftmost derivation of $w_b$ has the form
\[S {\overset{*}{\implies}} \alpha A \beta {\overset{*}{\implies}} \alpha x A y \beta {\overset{*}{\implies}} w_b
    \text{\quad such that $xy$ contains a $b$. }\]
Thus $A$ is a variable of Type~\ref{onlyb} or Type~\ref{ab}
(note that we have already ruled out Type~\ref{bc} above).
More precisely, we have that $x$ belongs to $a^+$ or $b^*$ and $y = b^j$ for some $j \in \nats$.
Now we show that the case where $x$ belongs to $a^+$ is not possible.
Assume for contradiction that $x = a^i$ for some $i>0$.
Then we also have the derivation
\[S {\overset{*}{\implies}} \alpha A \beta {\overset{*}{\implies}} a^{2p-i}b^{p-j}c^{2p} \notin B.\]
Therefore, $A$ is a Type~\ref{onlyb} variable that is used in the derivation of $w_b$, which can be used to inject another $b^p$, yielding a derivation of $w$. Thus, we have exhibited a derivation of $w$ that uses a variable of Type~\ref{ac}.

Now consider a derivation of the word $w_a = a^{p}b^{2p}c^{2p}$.
Such a derivation cannot contain a variable of Type~\ref{ac} since this allows us either to produce a word that is not in $a^*b^*c^*$ or to inject $a^{p}c^{r}$ for $p \geqslant r > 0$, leading to the derivation of $a^{2p}b^{2p}c^{2p+r} \notin B$.
Further, arguing as above, some variable of Type~\ref{onlya}
must appear in the derivation of $w_a$ that is used to obtain sufficient number of $a$'s in the derivation of $w_a$.
Such a variable of Type~\ref{onlya} can be used to inject $a^p$ into $w_a$ which leads to the derivation of $w$.
Thus, we have exhibited a derivation of $w$ that does not contain a variable of Type~\ref{ac}.

Altogether, there are two different leftmost derivations of the word $w$.
Thus, $G$ is not unambiguous, yielding the desired contradiction.
\end{proof}

\section{Succinctness}
\label{section_succinctness}

We show that \HD-PDA are not only more expressive 
than DPDA, but also more succinct. Similarly, we show that PDA are more succinct than \HD-PDA. Recall that the size of a \pda{} is the sum of its state-space and stack alphabet.

\begin{thm}\label{theorem:succinctness}
\HD-PDA can be exponentially more succinct than DPDA,
and PDA can be
double-exponentially more succinct than \HD-PDA.
\end{thm}

We first show that \HD-PDA can be 
exponentially more succinct than DPDA.
To this end, we construct a family $(C_n)_{n \in \nats}$
of languages such that $C_n$ is recognised by an \HD-PDA
of size $O(n)$,
yet every DPDA recognising $C_n$
has at least exponential size in $n$.

Let $c_n\in (\$\{0,1\}^n)^*$ be the word describing
an $n$-bit binary counter counting from $0$ to $2^{n}-1$.
For example,
$c_2=\$00\$01\$10\$11$.
We consider the family of languages~$ 
C_n= \big\{w \in \{0,1,\$,\#\}^*\mid w \neq c_n\#\big\} \subseteq \{0,1,\$,\#\}^*$ of bad counters. Each $C_n$ is recognised by a DPDA since it is co-finite.

We show that the language $C_n$ is recognised by
an \gfgpda{} of size $O(n)$ and that every DPDA~$\daut$ recognising $C_n$ has exponential size in $n$. 
Observe that this result implies that even \HD-PDA that are equivalent to DPDA are not determinisable by pruning.
In contrast, for NFA, \HDness implies determinisability by pruning~\cite{BKS17}.

\begin{lem}
\label{lemma:succinct_GFGPDA1}
The language $C_n$ is recognised by
an \gfgpda{} of size $O(n)$.
\end{lem}

\begin{proof}
We define a PDA~$\aut=(Q, \Sigma, \Gamma, q_\initmark, \Delta, F)$ that recognises $C_n$.
The automaton $\aut$
operates in three phases:
a push phase,
followed by a check phase,
and then a final phase.
These phases work as follows.
Suppose that $\aut$ receives an input $w \in \{0,1,\$,\#\}^*$.
During the first phase, $\aut$ pushes the input processed onto the stack until the sequence $1^n$ appears.
If it never appears, the input is accepted.
During this phase, $\aut$ also checks whether the prefix $w'$ of $w$
processed up to this point
is a sequence of counter values starting with $0^n$,
i.e. whether $w'$ is in the language
\[
L_c=
\{
\$d_0\$ d_1\$ \cdots \$ d_m \mid
 d_0 = 0^n,
 d_i = 1^n \Leftrightarrow i = m,
\textup{ and }
 d_i \in \{0,1\}^n
 \textup{ for all }
 1 \leqslant i \leqslant m
 \}.
\]
If $w' \not\in L_c$,
then $\aut$ immediately accepts.
Otherwise, $\aut$ moves to the second phase.
During the check phase, $\aut$ pops the stack.
At any point, $\aut$ can nondeterministically guess that the top symbol of the stack
is evidence of bad counting (details are described below). It then accepts the input if the guess was correct.
If $\aut$ completely pops the stack without
correctly guessing an error in the counter,
it moves to the final phase.
Since the prefix $w'$ processed up to this point
ends with the sequence $1^n$,
if $\aut$ now processes any suffix different from a single $\#$,
then the input is not equal to $c_n\#$, and can be accepted.

The stack alphabet of $\aut$ has constant size $3$.
The push phase requires $3(n+1)$ states:
\begin{itemize}
    \item First, $\aut$ checks whether $\$ 0^n$ is a prefix of the input.
    This can be done with $n + 2$ states.
    \item Then, $\aut$ checks whether the following $\{0,1\}^*$ segments are $n$-bits wide,
    and only the last one is $1^n$.
    This can be done with $2n+1$ additional states:
    repeatedly, $\aut$ processes $n+1$ symbols,
    checks whether only the first of them is a $\$$,
    and keeps track of whether at least one of them is $0$.
\end{itemize}

We now show that $6(n+1)$ additional states
are enough for the check phase.
To this end, we study the errors that $\aut$ needs to check.
Note that, to increment the counter correctly,
we need to change the value of all the bits starting from the last $0$,
and leave the previous bits unchanged.
Therefore, $\aut$ can recognise with $6(n+1)$ states
whether the top symbol of the stack does not correspond
to a correct counter increment:
$\aut$ pops the top $n+1$ stack symbols while keeping in memory
\begin{itemize}
    \item the value of the first symbol popped;
    \item whether we have not yet popped a $\$$ (there is exactly one $\$$ in the top $n+1$ stack symbols, as the stack content is in $L_c$),
    or a $\$$ but no $0$ afterwards,
    or a $\$$ and at least one $0$ afterwards.
\end{itemize}
The input is accepted whenever the first symbol popped and the top stack symbol after popping match
yet no $0$ has been popped between the $\$$ and the last symbol,
or they differ yet at least one $0$ has been popped
between the $\$$ and the last symbol.

Finally, only three states are needed for the final phase:
when the bottom of the stack is reached, $\aut$
transitions to a new state, and from there it checks
whether the suffix is in the language
$\{0,1,\$,\# \}^* \setminus \{ \# \}$.

To conclude, note that $\aut$ is history-deterministic:
the only nondeterministic choice happens during the check phase,
and the resolver knows which symbols of the stack are evidence of bad counting.
Note that this choice only depends
on the current stack content.
\end{proof}

As mentioned above, each $C_n$ is recognised by a DPDA; now we show that such a DPDA must be large.
\begin{lem}
\label{lemma:succinct_GFGPDA2}
Every DPDA recognising the language $C_n$ has
at least exponential size in $n$.
\end{lem}

\begin{proof}
It is known that every DPDA can be complemented
at the cost of multiplying
its number of states by three~\cite{Hopcroft}.
Therefore, to prove the statement,
we show that even every PDA recognising 
the complement $\{c_n\#\}$ of $C_n$
has at least exponential size in $n$:

\begin{clm}
Every PDA $\aut=(Q, \Sigma, \Gamma, q_\initmark, \Delta, F)$
recognising $\{ c_n\# \}$ has a size greater than $2^{(n-1)/3}$.
\end{clm}
To prove the claim, we transform $\aut$ into a context-free
grammar generating the singleton language $\{ c_n\# \}$,
and then we show that such a grammar requires
exponentially many variables.
This is a direct consequence of the $mk$ Lemma~\cite{DBLP:journals/tit/CharikarLLPPSS05},
but proving it directly using similar techniques yields a slightly better bound.

Before changing $\aut$ into a grammar,
we slightly modify
its acceptance condition:
we add to $\aut$ a fresh final state $f$
in which the stack can be completely popped
\emph{including the bottom of stack symbol} $\bot$
(which normally cannot be touched
according to our definition of PDA).
Moreover, we allow $\aut$ to transition
towards $f$ nondeterministically
from all of its other final states.
This new automaton, which accepts by empty stack,
is easily transformed into 
a grammar $\mathcal{G}$ using
the standard transformation~\cite{Hopcroft}:
\begin{itemize}
    \item The terminals of $\mathcal{G}$ are $0$, $1$, $\$$ and $\#$.
    \item The variables of $\mathcal{G}$ are the triples $(p,X,q)$,
    for every state $p, q \in Q \cup \{f\}$ and stack symbol
    $X \in \Gamma_{\bot}$.
    \item The initial variable is $(q_\initmark, \bot, f)$,
    where $q_\initmark$ is the initial state of $\aut$
    and $f$ is the fresh final state.
    \item
    Each transition $(p,X,a,q,\gamma) \in \Delta$ yields production rules as follows:
    \begin{enumerate}\item
        If $\gamma = \epsilon$, then $\mathcal{G}$
        has the production rule
        $(p,X,q) \rightarrow a$;
        \item
        If $\gamma = Y$,
        then
        $\mathcal{G}$ has the production rule $(p,X,q_1) \rightarrow a (q,Y,q_1)$
        for all $q_1 \in Q$;
        \item 
        If $\gamma = YZ$,
        then
        $\mathcal{G}$ has the production rule
        $(p,X,q_2) \rightarrow a (q,Y,q_1) (q_1,Z,q_2)$
        for all $q_1,q_2 \in Q$.
    \end{enumerate}
\end{itemize}
The variables can be interpreted as follows:
for every $p, q \in Q$ and $X \in \Gamma$,
the variable $(p,X,q)$ can be derived into any input word $w \in \{0,1,\$, \#\}^*$
that $\aut$ can process starting in state $p$ and ending in state $q$
while consuming the symbol $X$ from the top of the stack.
Therefore, in particular, 
since the initial variable is $(q_\initmark, \bot, f)$,
$\mathcal{G}$ generates the same language as $\aut$.

We now prove that the grammar $\mathcal{G}$
has at least $2^{n-1}$ distinct variables,
hence $(|Q|+1)^2(|\Gamma|+1) \geqslant 2^{n-1}$,
which implies that the size
$|Q| + |\Gamma|$ of $\aut$ is at least $\geqslant 2^{(n-1)/3}$.
To this end, we study a (directed and ordered) derivation tree $T$ of the word $c_n\#$.

Remember that $c_n = \$d_0\$d_1\$ \cdots \$ d_{2^{n}-1}$ represents
an $n$-bit binary counter counting from $0$ to $2^n-1$.
For each $0 \leq i \leq 2^n-1$,
let us consider the vertex $v_i$ of $T$
such that the counter value $d_i$ is an infix of the derivation of $v_i$,
but of none of its children.
In other words, $d_i$ is split between the derivations of the children of $v_i$.
By definition of the grammar $\mathcal{G}$,
each vertex of $T$ has at most three children,
hence at most two counter values can be split amongst the children of a given vertex.
Therefore,
for every $1 \leq i \leq 2^{n}-2$,
while $v_i = v_{i-1}$ or $v_i = v_{i+1}$ might hold,
we know that $v_i \neq v_{j}$ for each $j \in \{1,2,\ldots,2^n-1\} \setminus \{i-1,i,i+1\}$. 
As a consequence, the vertices $v_0,v_2,v_4, \ldots, v_{2^n-2}$ are all distinct.
Finally, since $c_n\#$ is the only word recognised by $\aut$
and each counter value $d \in \{0,1\}^n$ appears a single time as an infix of $c_n\#$,
the $2^{n-1}$ variables labelling these vertices need to be distinct.
\end{proof}

Now, we consider the gap between \HD-PDA and PDA.
We show that there exists a family $(L_n)_{n > 0}$ of languages such that $L_n$ is recognised by a PDA of size $O(\log n)$ while every \HD-PDA recognising this language has at least exponential size in $n$.

Formally, we set
$L_n =  (0+1)^*1(0+1)^{n-1}$, that is, the $n^{\text{th}}$ bit from the end is a $1$. Here, we count starting from $1$, so that the last bit is the $1^{\text{st}}$ bit from the end.
Note that this is the standard example for showing that NFA can be exponentially more succinct than DFA, and has been used for many other succinctness results ever since.

\begin{lem}
\label{lem:PDA-logn}
There exists a PDA of size $O(\log n)$ recognising $L_n$.
\end{lem}

\begin{proof}
We describe a PDA $\aut_n$ that recognises $L_n$.
The PDA $\aut_n$ nondeterministically guesses the $n^\text{th}$ bit from the end, checks that it is a $1$ and switches to a counting gadget that checks that the word ends in $n$ steps, as follows:
\begin{enumerate}[(i)]
    \item It pushes the binary representation of $n-2$ onto the stack.
    For example, if $n=8$, then $110$ is pushed onto the stack with $0$ at the top.
    Note that $\log (n-2)$ states suffice for pushing the binary representation of $n-2$.
    If $n=1$, then instead of pushing anything onto the stack, the automaton directly moves to a final state without any enabled transitions.
    \item \label{step:counter} Then $\aut_n$ moves to a state that attempts to decrement the counter by one for each successive input letter, as follows:
    When an input letter is processed, it pops $0$'s until $1$ is at the top of the stack, say $m$ $0$'s. Then, it replaces the $1$ with a $0$, and finally pushes $m$ $1$'s back onto the stack before processing the next letter.
     If the stack empties before a $1$ is at the top of the stack, then the counter value is $0$ and the automaton moves to  a final state with no enabled transitions.
    Note that $O(\log n)$ states again suffice for this step.
\end{enumerate}
Thus, $\aut_n$ has $O(\log n)$ states.
Note that for all $n$, $\aut_n$ uses only three stack symbols, $0$, $1$, and $\bot$.
Altogether, the size of $\aut_n$ is $O(\log n)$, and $\aut_n$ recognises $L_n$.
\end{proof}

\begin{lem}
\label{lem:GFG-PDA-exponential}
Every \HD-PDA recognising $L_n$ has at least exponential size in $n$.
\end{lem}

 Towards proving this, we define the following notions.
We say that a word $w$ of length $n$ is \emph{rotationally equivalent} to a word $w'$ if $w'$ is obtained from $w$ by rotating it.
For example, the word $w=1101$ is rotationally equivalent to  $w'=1110$ since $w'$ can be obtained from $w$ by rotating it once to the right.
Note that the words that are rotationally equivalent form an equivalence class, and thus rotational equivalence partitions $\{0,1\}^n$.
Since the size of each class is at most $n$, the number of equivalence classes is at least $\frac{2^n}{n}$.

Now, we define the \emph{stack height} of a configuration~$c= (q, \gamma)$ as $\sh(c) = \size{\gamma}-1$,
and we define \emph{steps} of a run as usual:
Consider a run~$c_0 \tau_0 c_1 \tau_1 \cdots c_{n-1} \tau_{n-1} c_n$.
A position~$s \in \set{0,\ldots, n}$ is a step if for all $s' \geqslant s$, we have that $\sh(c_{s'})\geqslant \sh(c_s)$, that is, the stack height is always at least $\sh(c_s)$ after position $s$.
Any infinite run of a PDA has infinitely many steps.
We have the following observation.
\begin{rem} \label{rem:steps}
If two runs of a PDA have steps $s_0$ and $s_1$, respectively, with the same mode, then the suffix of the run following the step $s_0$ can replace the suffix of the other run following the step $s_1$, and the resulting run is a valid run of the PDA.
\end{rem}

Now, we are ready to prove Lemma~\ref{lem:GFG-PDA-exponential}. Here, we work with infinite inputs for \HD-PDA. The run induced by a resolver on such an input is the limit of the runs on the prefixes.

\begin{proof}
Let $\aut$ be an \HD-PDA with resolver~$r$ that recognises $L_n$.  We show that $|Q| \cdot |\Gamma| \geqslant \frac{2^n}{n}$, where $Q$ is the set of states and $\Gamma$ is the stack alphabet of $\aut$.

Towards a contradiction, assume that $|Q|\cdot |\Gamma| < \frac{2^n}{n}$.
Then there exist two words $w_0$ and $w_1$ of length~$n$ that are not rotationally equivalent and such that the runs~$\rho_0$ and $\rho_1$ of $\aut$ induced by $\hstrat$ on $w_0^\omega$ and $w_1^\omega$ contain steps with the same mode, at positions $s_0$ and $s_1$ in $\rho_0$ and $\rho_1$ respectively, such that at least $n$ letters are processed  before $s_0$ and $s_1$.
Now consider in each of these two runs the sequence of input letters of length $n$
preceding and including the step position.
Let these $n$-letter words be $w_0'$ and $w_1'$ respectively.
Since $w_0$ and $w_1$ are not rotationally equivalent, $w_0'$ and $w_1'$ differ in at least one position $j \leqslant n$.

W.l.o.g., assume that for $w_0'$, the bit at position $j$ is $0$, while it is $1$ at position $j$ for $w_1'$.
Since the resolver chooses a run such that for every word where the $n^\text{th}$ letter from the end is a $1$ is accepted, this implies that $\rho_0$ does not visit a final state after processing $j-1$ letters after $s_0$, while $\rho_1$ visits a final state after processing $j-1$ letters after $s_1$.

Now we reach a contradiction as follows.
The suffix of $\rho_0$ starting from position $s_0+1$ can be replaced with the suffix of $\rho_1$ starting from position $s_1+1$. 
By Remark~\ref{rem:steps}, this yields a valid run~$\rho$ of $\aut$.
However, since the state that occurs after $j-1$ letters are processed after position $s_1$ in $\rho_1$ is final, after the replacement, the state that occurs after $j-1$ letters are processed after position $s_0$ in $\rho$ is final as well.
However the $n^\text{th}$ letter from the end of the word processed by this accepting run of $\aut$ is a $0$, contradicting that $\aut$ recognises $L_n$.

Thus, we have that $|Q| \cdot |\Gamma|$ is at least equal to the number of rotationally equivalent classes, that is, $|Q| \cdot |\Gamma| \geqslant \frac{2^n}{n}$.
Hence, the size of $\aut$ is at least $\frac{2^{\frac{n}{2}}}{\sqrt{n}}$.
\end{proof}



\section{History-deterministic Visibly Pushdown Automata}
\label{section_vpa}
As we shall see in Section~\ref{sec_closure}, one downside of \HD-PDA is that, like $\omega$-\HD-PDA, they have poor closure properties and checking \HDness is undecidable. We therefore consider a well-behaved class of \HD-PDA, namely \HD visibly pushdown automata, \HD-VPA for short, that is closed under union, intersection, and complementation.

Let $\Sigma_c, \Sigma_r$ and $\Sigma_{\sf int}$ be three disjoint sets of \emph{call} symbols, \emph{return} symbols and \emph{internal} symbols respectively.
Let $\Sigma = \Sigma_c \cup \Sigma_r \cup \Sigma_{\sf int}$.
A \emph{visibly pushdown automaton} \cite{AlurM04} (VPA) $\aut=(Q, \Sigma, \Gamma, q_I, \Delta, F)$ is a restricted PDA that pushes onto the stack only when it reads a call symbol, it pops the stack only when a return symbol is read, and does not use the stack when there is an internal symbol.
Formally,
\begin{itemize}
    \item A letter~$a \in \Sigma_{\sf c}$ is only processed by transitions of the form~$(q, X, a, q', XY)$ with $X\in \Gammabot$, i.e. some stack symbol~$Y \in \Gamma$ is pushed onto the stack.
    
    \item A letter~$a \in \Sigma_{\sf r}$ is only processed by transitions of the form~$(q, X, a, q', \epsilon)$ with $X \neq \bot$ or $(q, \bot, a, q',\bot)$, i.e. the topmost stack symbol is removed, or if the stack is empty, it is left unchanged.
    
    \item A letter~$a \in \Sigma_{\sf int}$ is only processed by transitions of the form~$(q, X, a, q',X)$ with $X \in \Gammabot$, i.e. the stack is left unchanged.

    \item There are no $\epsilon$-transitions.
    
\end{itemize}
Intuitively, the stack height of the last configuration of a run processing some $w \in (\Sigmacall \cup \Sigmareturn \cup \Sigmaskip)^*$ only depends on $w$.

We denote by \HD-VPA the VPA that are history-deterministic.
Every VPA (and hence every \HD-VPA) can be determinised, i.e. all three classes of automata recognise the same class of languages, denoted by \vpl, which is a strict subset of \dcfl~\cite{AlurM04}.

\subsection{Succinctness}

While all three classes of VPA are equally expressive, VPA can be exponentially more succinct than deterministic VPA (DVPA)~\cite{AlurM04}. We show that there is an exponential gap both between the succinctness of \HD-VPA and DVPA and between VPA and \HD-VPA. The proof of the former gap again uses a language of bad counters, similar to $C_n$ used in Theorem~\ref{theorem:succinctness}, which we adapt for the VPA setting by adding a suffix allowing the automaton to pop the stack. Furthermore, for the gap between VPA and \HD-VPA, we similarly adapt the language $L_n$ of words where the $n^{\text{th}}$ bit from the end is a $1$, from the proof of Theorem~\ref{theorem:succinctness}, by making sure that the stack height is always bounded by $1$. Such an \HD-VPA is essentially an \HD-NFA, and therefore determinisable by pruning, which means that it is as big as a deterministic automaton for the language.

\begin{thm}\label{theorem:succinct_VPA}
\HD-VPA can be exponentially more succinct than DVPA and VPA can be exponentially more succinct than \HD-VPA.
\end{thm}
We split the proof into two parts.

\begin{lem}
\HD-VPA can be exponentially more succinct than DVPA.
\end{lem}
\begin{proof}
We construct a family~$(C_n')_{n \in \nats}$ of languages
such that there is an \HD-VPA
of size~$O(n)$ recognising $C_n'$,
yet every DVPA recognising $C_n'$ has at least exponential size in $n$.
This family is obtained by adapting
the family $(C_n)_{n \in\nats}$ that we used to prove
the succinctness of \HD-PDA in Section \ref{section_succinctness}:
Once again, we consider
the word $c_n\in (\$\{0,1\}^n)^*$ describing
an $n$-bit binary counter counting from $0$ to $2^{n}-1$.
We consider the languages
\[
C_n'= \big\{w \in \{0,1,\$,\#\}^*\mid w \neq c_n\#^{2^n(n+1)}\big\} \subseteq \{0,1,\$,\#\}^*
\]
 of bad counters, where $0$, $1$ and $\$$ are call symbols and
$\#$ is a return symbol.
The only difference with $C_n$
is that the forbidden word is $c_n\#^{2^n(n+1)}$
instead of $c_n\#$.
An \HD-VPA of size $O(n)$ recognising $C_n'$
is obtained by a small modification of the construction
presented in the proof of 
Lemma \ref{lemma:succinct_GFGPDA1}.
We adapt the construction of the automaton~$\aut$
recognising $C_n$ as follows:
\begin{itemize}
    \item 
    The push phase is identical;
    \item 
    The check phase is performed
    by consuming the $\#$ symbols
    instead of having $\epsilon$-transitions.
    While the stack is not empty,
    $\aut$ accepts even if it has not
    found evidence of bad counting yet.
    Moreover, $\aut$ transitions towards a final sink state
    if a non-$\#$ symbol is read.
    Once the stack is empty, it transitions towards
    the final phase;
    \item 
    In the final phase,
    since the prefix processed up to this point
    ends with an empty stack,
    if the suffix left to read is non-empty
    then the input is not equal to $c_n\#^{2^n(n+1)}$,
    and can be accepted.
\end{itemize}

Finally, we can prove 
that every DPDA (and in particular every DVPA)
recognising $C_n'$
has at least exponential size in $n$
in the exact same way as we proved
Lemma~\ref{lemma:succinct_GFGPDA2}:
The proof only uses the fact that $c_n$
is an infix of the single word rejected by $C_n$, 
thus $C_n'$ can be treated identically.
Note that this lower bound is independent
of the partition of the letters
into calls, returns, and internals.
\end{proof}

\begin{lem}
\label{lemma_vpavsgfgvpa}
VPA can be exponentially more succinct than \HD-VPA.
\end{lem}

We show that there exists a family~$(L'_n)_{n \in \nats}$ of languages such that there exists a VPA of size~$O(n)$  recognising $L_n'$ while every \HD-VPA recognising the same language has size at least $2^{n/6}$.

Towards this we consider a language $L'_n$ of words in $ (01+10)^* \cdot (\epsilon+0+1)$ with the $n^\text{th}$ last letter being $1$.
We first note that $L'_n$ can be recognised by a  VPA with $O(n)$ states, which checks that the input is in $(01+10)^* \cdot (\epsilon+0+1)$ and nondeterministically guesses the $n^{\text{th}}$ last letter and verifies that it is a $1$.

First, we note that every DFA recognising $L'_n$ has exponential size, which can be shown by counting the equivalence classes of the Myhill-Nerode congruence of $L'_n$ (see, e.g.~\cite{Hopcroft}).

\begin{rem} \label{rem:GFG_Ln_DFA}
Every DFA recognising $L'_n$ has at least $2^{\lceil n /2 \rceil}$ states.
\end{rem}

Using this, we obtain an exponential lower bound on the size of \HD-VPA recognising $L_n'$, thereby completing the proof of Lemma~\ref{lemma_vpavsgfgvpa}.

\begin{lem}
Every \HD-VPA recognising $L'_n$ has at least size~$2^{\lceil n/6 \rceil}$.
\end{lem}

\begin{proof}
The proof is based on the fact that \HD-NFA can be determinised by pruning~\cite{BKS17}, that is, they always contain an equivalent DFA, i.e. the lower bound of Remark~\ref{rem:GFG_Ln_DFA} is applicable to \HD-NFA as well.

Let $\aut$ be an \HD-VPA recognising $L_n'$.
We consider the following cases:
\begin{enumerate}
    \item \emph{Both $0$ and $1$ are either a return symbol or an internal symbol}: The \HD-VPA $\aut$ in this case can essentially be seen as an \HD-NFA with the same set of states, since the stack is not used (it is always equal to $\bot$).
    Given that \HD-NFA are determinisable by pruning, by Remark~\ref{rem:GFG_Ln_DFA} such an \HD-NFA has at least $2^{\lceil n/2 \rceil}$ states.
    \item \emph{At least one of $0$ and $1$ is a call symbol while the other one is a call or an internal symbol}: Let $Q$ be the set of states and $\Gamma$ be the stack alphabet of $\aut$.
    Since the height of the stack is nondecreasing, $\aut$ has only access to the top stack symbol.
    We can thus construct an equivalent \HD-NFA over finite words with states in $Q \times \Gamma$.
    Since \HD-NFA are determinisable by pruning, and using Remark \ref{rem:GFG_Ln_DFA} again, we have that $|Q| \cdot |\Gamma| \geqslant 2^{\lceil n/2 \rceil}$.
    Thus either $|Q| \geqslant 2^{n/4}$ or $|\Gamma| \geqslant 2^{n/4}$.
    Hence for this case, we have that the size of the \HD-VPA is at least $2^{n/4}$.
    \item \emph{One of $0$ and $1$ is a call symbol while the other one is a return symbol}: Note that since a word in $L_n'$ is composed of sequences of $10$ and $01$, the stack height can always be restricted to $2$.
    Thus the configuration space of $\aut$, restricted to configurations on accepting runs, is finite, and there is an equivalent \HD-NFA of size at most $|Q| \cdot |\Gamma|^2$.
    Thus $|Q| \cdot |\Gamma|^2 \geqslant 2^{\lceil n/2 \rceil}$ giving either $|Q| \geqslant 2^{\lceil n/6 \rceil}$ or $|\Gamma| \geqslant 2^{\lceil n/6 \rceil}$.
    Again by the determinizability by pruning argument, we have that the size of the \HD-VPA $\aut$ is at least $2^{\lceil n/6 \rceil}$.
\qedhere
\end{enumerate}
\end{proof}

\subsection{Deciding HDness of VPA}

We now turn to the question of deciding whether a given VPA is \HD. We show decidability  using the \textit{one-token game}, introduced by Bagnol and Kuperberg~\cite{BK18}. It is easier to decide than the game-based characterisation of \HDness of $\omega$-regular automata by Henzinger and Piterman~\cite{HP06}. While the one-token game does not characterise the \HDness of B\"uchi automata~\cite{BK18}, here we show that it suffices for VPA.

\begin{thm}
\label{thm:vpagfgnesscomplexity}
The following problem is $\exptime$-complete: Given a VPA~$\aut$, is $\aut$ \HD?
\end{thm}

\begin{proof}
We first define the \textit{one-token game}, introduced by Bagnol and Kuperberg~\cite{BK18} in the context of regular languages, for VPA. Fix a VPA $\aut= (Q,\Sigma,\Gamma,q_I,\Delta,F)$, which we assume w.l.o.g.\ to be complete, i.e. every mode has at least one enabled $a$-transition for every $a \in \Sigma$.
The positions of the one-token game consist of pairs of configurations $(c,c')$, starting from the pair containing the initial configuration of $\aut$ twice. At each round $i$ from position $(c_i, c'_i)$:
\begin{itemize}
    \item Player~1 picks a letter $a_i\in \Sigma$, then
    \item Player~2 picks an $a_i$-transition $\tau_i \in \Delta$ enabled in $c_i$, leading to a configuration $c_{i+1}$, then
    \item Player~1 picks an $a_i$-transition $\tau'_i\in \Delta$ enabled in $c_i'$, leading to a configuration $c_{i+1}'$.
    \item Then round $i$ ends. In round $i+1$, the play then proceeds from the position $(c_{i+1},c_{i+1}')$.
\end{itemize}
Note that due to completeness of $\aut$, there is always at least one move available for each player. 

The moves of the two players during a play induce an infinite word $a_0a_1\cdots\in \Sigma^\omega$ and two runs~$c_0\tau_0c_1\tau_1c_2\cdots$ and $c_0'\tau'_0c_1'\tau'_1c_2'\cdots$ built by Player~2 and Player~1 respectively. Player~2 wins if for all $n \geqslant 0$, whenever the run $c_0'\tau'_0\cdots\tau'_n c_{n+1}'$ constructed by Player~1 is accepting then the run~$c_0\tau_0\cdots \tau_nc_{n+1}$ constructed by Player~2 is accepting as well. Recall that VPA do not have $\varepsilon$-transitions, so the two runs proceed in lockstep, i.e. both runs process  $a_0\cdots a_n$.

Observe that this game can be seen as a safety game on an infinite arena induced by the configuration graph of a visibly pushdown automaton, obtained by taking the product of $\aut$ with itself. 
This in turn is solvable in exponential time~\cite{DBLP:journals/iandc/Walukiewicz01}. 

It now suffices to argue that this game characterises whether the VPA $\aut$ is \HD:
We argue that $\aut$ is \HD if and only if Player~2 wins the one-token game on $\aut$.
One direction is immediate: if $\aut$ is \HD, then the resolver induces a strategy for Player~2 in the one-token game which ensures that the run constructed by Player~2 is accepting whenever Player~1 picks a word that is accepted by $\aut$. 
This covers in particular the cases when the run constructed by Player~1 is accepting, which suffices for Player~2 to win.

For the converse direction, we show that a winning strategy for Player~2 in the one-token game can be turned into a resolver for $\aut$.
To this end, consider the family of \textit{copycat strategies} for Player~1 that copy the transition chosen by Player~2 until she plays an $a$-transition from a configuration~$c$ to a configuration~$c'$ such that there is a word~$aw$ that is accepted from $c$ but $w$ is not accepted from $c'$. We call such transitions non-residual. If Player~2 plays such a non-residual transition, then the copycat strategies stop copying and instead play the letters of $w$ and the transitions of an accepting run over $aw$ from $c$. 

If Player~2 wins the one-token game with a strategy $\strat$, she wins, in particular, against every copycat strategy for Player~1. Observe that copycat strategies win any play along which Player~2 plays a non-residual transition. Therefore, $\strat$ must avoid ever playing a non-residual transition.
We can now use $\strat$ to construct a resolver~$r_\strat$ for $\aut$: $r_\strat$ maps a sequence of transitions over a word $w$ and a letter~$a$ to the transition chosen by $\strat$ in the one-token game where Player~1 played $wa$ and used a copycat strategy to construct his run. Then, $r_\strat$ never produces a non-residual transition. As a result, if a word $w$ is in $L(\aut)$, then the run induced by $r_\strat$ over every prefix $v$ of $w$ leads to a configuration that accepts the remainder of $w$. This is in particular the case for $w$ itself, for which $r_\strat$ induces an accepting run. This concludes our argument that $r_\strat$ is indeed a resolver, and $\aut$ is therefore \HD.

Thus, to decide whether a VPA $\aut$ is \HD it suffices to solve the one-token game on $\aut$, which can be done in exponential time.
The matching lower bound follows from a reduction from the inclusion problem for VPA, which is $\exptime$-hard~\cite{AlurM04}, to \HDness (see~\cite[Theorem~6.1(1)]{LZ22} for details of the reduction in the context of $\omega$-\HD-PDA).
\end{proof}

Thus, \HDness for VPA is decidable. 
On the other hand, the problem of deciding whether a language of a VPA is history-deterministic, i.e. accepted by some \HD-VPA, is trivial, as VPA are determinisable~\cite{AlurM04}. 

\subsection{Good-enough Synthesis}

Finally, we relate the \HDness problem to the \textit{good-enough synthesis} problem~\cite{AK20}, also known as the \textit{uniformization} problem~\cite{carayol:hal-01806575}, which is similar to Church's synthesis problem\footnote{Compare the following definitions to that of Gale-Stewart games in Section~\ref{sec_games.tex}, which capture Church's synthesis problem.}, except that the system is only required to satisfy the specification on inputs in the projection of the specification on the first component.

Let $w \in \Sigma_1$ and $w' \in \Sigma_2$ with $|w| = |w'|$. Then, for the sake of readability, we write $\binom{w}{w'}$ for the word~$\binom{w(0)}{w'(0)} \cdots \binom{w(|w|-1)}{w'(|w|-1)}$ over $\Sigma_1 \times \Sigma_2$.

\begin{defi}[\good-synthesis]
Given a language $L\subseteq (\Sigma_1\times \Sigma_2)^*$, is there a function $f:\Sigma_1^* \rightarrow \Sigma_2$ such that for each $w\in \{w\mid \exists w'\in \Sigma_2^*. \binom{w}{w'})\in L\}$ the word $\binom{w}{w'}$ is in $L$, where $w'(n) = f(w(0) \cdots w(n))$ for each $0 \le n < |w|$.
\end{defi}

We now prove that the \good-synthesis problem for \HD-VPA and DVPA is as hard as the \HDness problem for VPA, giving us the following corollary of Theorem~\ref{thm:vpagfgnesscomplexity}.

\begin{cor}\label{cor:good-enough}
The \good-synthesis problem for inputs given by \HD-VPA or DVPA is \exptime-complete.
\end{cor}

\begin{proof}
We show that deciding \good-synthesis and \HDness, which is \exptime-complete (see Theorem~\ref{thm:vpagfgnesscomplexity}), are polynomially equivalent. We show the upper bound for \HD-VPA and the lower bound for DVPA.

We first reduce the good-enough synthesis problem to the \HDness problem.
Given an \HD-VPA $\aut= (Q,\Sigma_1\times\Sigma_2,\Gamma,q_I,\Delta,F)$, with resolver $r$, let $\aut'$ be $\aut$ projected onto the first component: $\aut'=(Q,\Sigma_1,\Gamma,q_I,\Delta',F)$ has the same states, stack alphabet and final states as $\aut$, but the transitions of $\aut'$ are obtained by replacing each label $\binom{a}{b}$ of a transition of $\aut$ by $a$.
Let each transition of $\aut'$ be annotated with a $\Sigma_2$-letter of a corresponding $\aut$-transition. Thus, $\aut'$ recognises the projection of $L(\aut)$ on the first component.

We show that $\aut$ has a \good-synthesis function if and only if  $\aut'$ is \HD.
First, a \good-synthesis function $f$ for $\aut$, combined with $r$, induces a resolver $r'$ for $\aut'$ by using $f$ to choose output letters and $r$ to choose which transition of $\aut$ to use; together these uniquely determine a transition in $\aut'$. Then, if $w\in L(\aut')$, $f$ guarantees that the annotation of the run induced by $r'$ in $\aut'$ is a witness $w'$ such that $\binom{w}{w'}\in \aut$, and then $r$ guarantees that the run is accepting, since the corresponding run in $\aut$ over $\binom{w}{w'}$ must be accepting.
Conversely, assume $\aut'$ is \HD. Then, a resolver for $\aut'$ induces a \good-synthesis function for $\aut$ by reading the $\Sigma_2$-annotation of the chosen transitions in $\aut'$. Indeed, the resolver produces an accepting run with annotation $w'$ of $\aut'$ for every word $w$ in the projection of $L(\aut)$ on the first component. The same run is an accepting run in $\aut$ over $\binom{w}{w'}$ which is therefore in $L(\aut)$. 

For the lower bound, we reduce the \HDness problem of a VPA $\aut=(Q,\Sigma,\Gamma,q_I,\Delta,F)$ to the \good-synthesis problem of a DVPA $\aut'=(Q,\Sigma\times\Delta,\Gamma,q_I,\Delta',F)$. The deterministic automaton $\aut'$ is as $\aut$ except that each transition $\tau \in \Delta$ over a letter $a \in \Sigma$ is replaced with the same transition over $\binom{a}{\tau}$ in $\Delta'$. In other words, $\aut'$ recognises the accepting runs of $\aut$ and its \good-synthesis problem asks whether there is a function that constructs on-the-fly an accepting run for every word in $L(\aut)$, that is, whether $\aut$ has a resolver.
\end{proof}

This result should be compared to the case of LTL specifications, for which the \good-synthesis problem is \twoexptime-complete~\cite{AK20}.

\section{Closure properties}
\label{sec_closure}
We show that, like $\omega$-\HD-PDA, \HD-PDA have poor closure properties, but we start with the positive results and consider closure under intersection, union, and set difference with regular languages.

\begin{thm}\label{thm:gfgcfl_reg}
If $L$ is in \gfgcfl and $R$ is regular, then $L \cup R$, $L \cap R$ and $L \backslash R$ are also in \gfgcfl, but $R \backslash L$ is not necessarily in \gfgcfl.
\end{thm}

\begin{proof}
Consider an \HD-PDA $\aut = (Q, \Sigma, \Gamma, q_\initmark, \Delta, F)$ recognising $L$, and a resolver~$\hstrat$ for $\aut$.
By definition, $\hstrat$ only has to induce a run on every $w \in L(\aut)$, but does not necessarily induce a run on $w \notin L(\aut)$. First, we turn $\aut$ into an equivalent \HD-PDA~$\aut'$ that has a resolver that induces a run on every input $w \in \Sigma^*$.
This property allows us then to take the product of $\aut'$ and a DFA for $R$.

To this end, we add a fresh nonfinal sink state~$q_s$ with a self-loop~$(q_s, X, a, q_s, X) $ for every input letter~$a\in \Sigma$ and every stack symbol~$X \in \Gammabot$. Also, we add transitions so that every configuration has, for every $a \in \Sigma$, an enabled $a$-transition to the sink. 
The resulting PDA~$\aut'$ is equivalent to $\aut$ and $\hstrat$ is still a resolver for it. But, we can also turn $\hstrat$ into a resolver~$\hstrat'$ that induces a run on every possible input as follows:
If $\ell(\tau_0 \cdots \tau_n)$ is a prefix of a word in $L(\aut)$, then we define $\hstrat'(\tau_0 \cdots \tau_n, a) = \hstrat(\tau_0 \cdots \tau_n, a)$ for every $a \in \Sigma$. Otherwise, if $\ell(\tau_0 \cdots \tau_n)$ is not a prefix of a word in $L(\aut)$, we define $\hstrat'(\tau_0 \cdots \tau_n, a) = (q, X, a, q_s, X)$, where $(q, X)$ is the mode of the last configuration of the run induced by $\tau_0 \cdots \tau_n$. 
Thus, as soon as the input can no longer be extended to a word in $L(\aut)$, the run induced by $\hstrat'$ moves to the sink state and processes the remaining input. 

Now, let $\auta$ be a DFA recognising $R$.
For $L \cup R$, we construct the product~PDA~$\aut_\cup$ of $\aut'$ and $\auta$ that simulates a run of $\aut'$ and the unique run of $\auta$ simultaneously on an input word and accepts if either the run of $\aut$ or the run of $\auta$ is accepting.
We note that when an $\epsilon$-transition is chosen in $\aut'$ by the resolver, then no move is made in $\auta$. 
As $\aut'$ has a run on every input, the product PDA~$\aut_\cup$ has one as well.

For $L \cap R$, we construct a PDA $\aut_\cap$ which is similar to $\aut_\cup$ with the difference that $\aut_\cap$ accepts when each of $\aut'$ and $\auta$ has an accepting run on the input word.

The PDA~$\aut_\cup$ and $\aut_\cap$ are both history-deterministic, since only the nondeterminism of $\aut'$ needs to be resolved.

Finally, since $L\backslash R = L \cap R^c$ and the complement~$R^c$ is regular, it follows that $L \backslash R$ is recognised by an \HD-PDA if there is an \HD-PDA recognising $L$.
For $R \backslash L$, note that $\Sigma^*$ is a regular language and $\Sigma^*\backslash L = L^c$. In Theorem~\ref{theorem_closure}, we show that \gfgcfl{} is not closed under complementation, which implies that $R \backslash L$ is not necessarily in \gfgcfl.
\end{proof}

Now, we consider the closure properties of \HD-PDA.

\begin{thm}
\label{theorem_closure}
\HD-PDA are not closed under union, intersection, complementation, set difference, concatenation, Kleene star and homomorphism.
\end{thm}
\begin{proof}
The proofs for union, intersection, complementation, and set difference are similar to those used for $\omega$-\HD-PDA~\cite[Theorem 5.1]{LZ22}. We state them here for completeness.

To show non-closure under union, consider the languages $L_1 = \{a^nb^n  \mid  n \geqslant 0\}$ and $L_2 = \{a^nb^{2n}  \mid  n \geqslant 0\}$ respectively. There exist a DPDA recognising $L_1$ and a DPDA recognising $L_2$.
Hence by Lemma~\ref{lemma_inclusions}, there also exist an \HD-PDA recognising $L_1$ and an \HD-PDA recognising $L_2$.
However, we show in Lemma \ref{lem:expressiveness1} that $L_1 \cup L_2$ is not recognised by any \HD-PDA.


For intersection, consider the languages $L_3 = \{a^nb^nc^m  \mid  m,n \geqslant 0\}$ and $L_4 = \{a^mb^nc^n  \mid  m,n \geqslant 0\}$.
There exist DPDA recognising $L_3$ and $L_4$.
Hence by Lemma \ref{lemma_inclusions} there exist \HD-PDA recognising $L_3$ and $L_4$.
Now let $L  = L_3 \cap L_4 = \{a^nb^nc^n  \mid  n \geqslant 0\}$.
As $L$ is not a \cfl  there does not exist any \HD-PDA recognising $L$.


For non-closure under complementation, recall the language
\begin{align*}
    B_2 &= \{a^i\$a^j\$b^k\$ \mid k \leqslant \max(i,j)\}.
\end{align*}
We prove in Lemma \ref{lem_expressiveness} that
$B_2 \in $ \gfgcfl,
yet Lemma~\ref{lem_not_DCFL} shows that its complement $B_2^c$
is not even a context-free language.

Closure under set difference implies closure under complementation since for every language $L$ over alphabet $\Sigma$, we have that the complement~$L^{c}$ is equal to $\Sigma^* \backslash L$.

Now we show that \gfgcfl is not closed under concatenation.
Consider again the languages $L_1 = \{a^n b^n \mid n \geqslant 0\}$, and $L_2 = \{a^n b^{2n} \mid n \geqslant 0\}$.
We showed that $L_1 \cup L_2$ is not an \gfgcfl.
Now let $L_5 = cL_1 \cup L_2$.
Clearly $L_5$ is a DCFL and hence an \gfgcfl since the initial letter determines whether the word is in $cL_1$ or in $L_2$.
Now the language represented by the regular expression $c^*$ is regular, and hence an \gfgcfl.
We argue that $c^*L_5$ is not an \gfgcfl, and hence \gfgcfl{}s are not closed under concatenation.

Assume for contradiction that $c^*L_5$ is an \gfgcfl.
Let $L_6 = c^*L_5 \:\cap \: ca^*b^* =cL_1 \cup cL_2$.
Since \gfgcfl{s} are closed under intersection with regular languages (see Theorem~\ref{thm:gfgcfl_reg}), $L_6$ is also an \gfgcfl.
Now if $L_6$ is an \gfgcfl, then there exists an \HD-PDA $\aut_6$ with a resolver $r$ such that for a word $w=cv \in L_6$ where $v \in L_1 \cup L_2$, the resolver $r$ after reading $c$ reaches some state $q$ of $\aut_6$ from which it induces an accepting run on $v$.
This allows us to construct from $\aut_6$ an \HD-PDA $\aut_{1 \cup 2}$ accepting $L_1 \cup L_2$ with $q$ as the initial state.
We reach a contradiction since $L_1 \cup L_2$ is not an \gfgcfl.

Now we show that \gfgcfl{s} are also not closed under Kleene star.
Again we consider the language $L_5$ above which is an \gfgcfl.
We show that $L_5^*$ is not an \gfgcfl.
Had $L_5^*$ been an \gfgcfl, this would imply that $L_5^* \: \cap \: ca^*b^*$ is an \gfgcfl.
However $L_5^* \: \cap \: ca^*b^* = cL_1 \cup cL_2$ which is the language $L_6$ defined above.
Now as we have already shown that the language $L_6$ is not an \gfgcfl, this implies that $L_5^*$ is also not an \gfgcfl.

For non-closure under homomorphism, consider the language
\begin{align*}
L ={} &{}
\left\{
\,\binom{a}{1}^n \binom{b}{\#}^n 
\,\middle|\, n\geqslant 0 \, \right\} \, \cup \,
\left\{
\,\binom{a}{2}^n \binom{b}{\#}^{2n} 
\,\middle|\, n\geqslant 0 \, \right\}
\end{align*}
is recognised by a DPDA, and hence by an \HD-PDA using Lemma \ref{lemma_inclusions}, but its projection (which is a homomorphism)
\[ 
\set{ a^n b^n \mid n\geqslant 0} \cup \set{ a^n b^{2n} \mid n\geqslant 0}
\]
cannot be recognised by an \HD-PDA (see Lemma~\ref{lem:expressiveness1}). 
\end{proof}

\section{Compositionality with games}\label{sec:composition}
One of the appeals of \hd~automata consists of the good compositional properties that they enjoy, which allows them to be used to solve  games. Namely, we consider arena-based games, in which two players, called Player~1 and Player~2, build a path, with the owner of the current position choosing an outgoing edge at each turn, and Player 2 wins if the label of the path is in the winning condition. If the winning condition is the language of a deterministic or \hd~automaton $\daut$ with acceptance condition $C$, taking the product (or composition) of the arena with $\daut$ yields an arena which, when treated as a game with winning condition $C$, has the same winner as the original game with winning condition $L(\daut)$. In short, this product construction reduces games with winning conditions recognised by deterministic or \hd~$C$-automata into $C$-games. For example, it can be used to reduce games with $\omega$-regular winning conditions, which are recognised by \hd{} and deterministic parity automata, into parity games. In this sense, \hd~ automata are just as good as deterministic automata when it comes to solving games via composition, making them appealing for applications such as reactive synthesis, which can be solved via games.

Here we first show that \gfgpda, like regular \hd~ automata, enjoy this compositional property, both with respect to finite and infinite arenas.
We then consider the converse question, namely, whether all \pda~that behave well with respect to composition, are \gfgpda. This is the case for ($\omega$-)regular automata~\cite{BL19}, but not for quantitative automata~\cite{BL21}. We show that for \pda~over finite words, those that enjoy compositionality with infinitely branching games are exactly those that are \hd. On the other hand, unlike for the ($\omega$-)regular setting, compositionality with finitely branching games does not suffice to guarantee that the automaton is \hd, even for automata over finite words.

Our proof of equivalence of compositionality with games and history-determinism for PDA over finite words uses the determinacy of  safety games. 
Generalising this proof to automata over infinite words would require determinacy of games with $\omega$-context-free winning conditions, which is a large cardinal assumption~\cite{Fin13}.

A $\Sigma$-arena $A=(V, V_1, V_2,E,\iota, \ell)$ consists of a directed graph~$(V, E)$ of which the positions~$V$ are partitioned into those belonging to Player~$1$, $V_1$, and those belonging to Player~$2$, $V_2$, of which the edges $E$ are labelled with $\Sigma \cup \set{\sink}$ via a labelling $\ell\colon E\rightarrow \Sigma\cup \{\sink\})$, and rooted at an initial position $\iota\in V$. We assume that all $\sink$-edges lead to a terminal position, that is,  without outgoing edges, that all other positions have at least one successor, and that all edges leading to a terminal position are labelled $\sink$. 

A play on $A$ is a finite path ending in a terminal position, or an infinite path. A strategy for a player is a mapping from finite paths ending in a position that belongs to that player to one of its outgoing edges. A play is consistent with Player~$i$'s strategy $\sigma$ if for all of its prefixes $\pi$ ending in $V_i$, the next edge is $\sigma(\pi)$.

An (arena-based) game~$(A, L)$ consists of a $\Sigma$-arena~$A$ and a language~$L\subseteq \Sigma^*$. 
A finite play is winning for Player~2 if it is labelled with a word in $L\cdot\sink$. 
Note that Player~$2$ loses in particular all infinite plays. 
Winning strategies are defined in the usual way, and we say that a player wins the game if they have a winning strategy.

We now consider the composition of an arena and an automaton, that is, a product game in which Player 2 not only has to guarantee that the label of the play in the arena is in the language of the automaton, but she also has to build an accepting run over that word, transition by transition, as the play in the arena progresses.

Formally, the positions of the product game $G(A,\aut)$ of a $\Sigma$-arena $A=(V,V_1,V_2,E,\iota, \ell)$ and a \pda~$\aut=(Q,\Sigma, \Gamma,q_I,\Delta,F)$ consists of a position in $V$ and a configuration of $\aut$, the initial position being $(\iota, q_I, \bot)$.
At each round~$i$ from position~$(v,q,\gamma)$, the player who owns the position~$v$ chooses an outgoing edge $(v,v')\in E$, labelled by some letter $\ell(v,v')=a$. If $a\neq \sink$, Player~$2$ then chooses a possibly empty sequence of $\varepsilon$-transitions followed by one $a$-transition, so that the combined sequence induces a finite run from  $(q,\gamma)$ leading to some $(q',\gamma')$.
Then, the play proceeds in round~$i+1$ from $(v', q',\gamma')$. If $a=\sink$, then the play ends in $(v',q,\gamma)$, as $v'$ is by assumption a terminal position of the arena. Player~$2$ wins by reaching a position~$(v,q,\gamma)$ where $v$ is the terminal position and $q$ is final. Again, strategies, winning strategies and winning the game are defined as expected.

Observe that $G(A,\aut)$ may have an infinite number of positions, and may have infinite branching, due to unbounded $\varepsilon$-transition sequences. If $A$ is finitely representable, so is $G(A,\aut)$. In particular, if $A$ is finite, then $G(A,\aut)$ is a reachability game on the configuration graph of a pushdown automaton.

\begin{defi}
Let $\aut$ be a \pda\ over an alphabet $\Sigma$.
\begin{itemize}
    \item $\aut$ is compositional if the following holds for all (finitely or infinitely branching) $\Sigma$-arenas~$A$: Player~$2$ wins the game~$(A,L(\aut))$ if and only if she wins $G(A,\aut)$.
    
    \item $\aut$ is weakly compositional if the following holds for all finitely branching $\Sigma$-arenas~$A$: Player~$2$ wins the game~$(A,L(\aut))$ if and only if she wins $G(A,\aut)$.
\end{itemize}
\end{defi}

Note that compositionality implies weak compositionality. 
In the remainder of this section, we compare history-determinism and (weak) compositionality. 
\begin{lem}
\label{lem:hdimpliescomposition}
Every \gfgpda\ is compositional. 
\end{lem}

\begin{proof}
Let $\aut$ be an \gfgpda\ over $\Sigma$ and let $A$ be a $\Sigma$-arena. 
If Player~$2$ wins the game $(A,L(\aut))$ with a strategy $\strat$, then she wins $G(A,\aut)$ by using $\strat$ in the $A$-component of $G(A,\aut)$ and the resolver $r$ that witnesses the \hdsm~of $\aut$ in the $\aut$-component of $G(A,\aut)$. Then, $\strat$ guarantees that the word built in $A$ is in $L(\aut)\cdot\sink$, and $r$ guarantees that Player~$2$ reaches a final state at the end of the word in $L(\aut)$, before $\sink$ is encountered.

Conversely, a winning strategy for Player~$2$ in $G(A,\aut)$ projected onto the $A$-component is a winning strategy for Player~$2$ in $(A, L(\aut))$.
\end{proof}

We now consider the converse: whether it is the case that all (weakly) compositional \pda\ are \gfgpda. In the regular setting, automata are weakly compositional if and only if they are compositional~\cite{BL19}. 
Here, however, we observe that the answer depends on whether infinite branching is permitted in the arenas considered. 

\begin{lem}\label{lem:composition-finite}
There exists a \pda\ that is weakly compositional, but not history-deterministic.
\end{lem}

\begin{proof}
Consider the 3-state \pda~in Fig.~\ref{fig:not-hd} that recognises $a^*b$, by first guessing an upper bound $m$ on the number of $a$'s in the word, then processing $a$ no more than $m$ times, followed by a $b$. 
The automaton accepts $a^nb$ by pushing $m\geq n$ elements onto the stack before reading $n$ times $a$, followed by $b$.

\begin{figure}\label{fig:not-hd}
    \centering
\begin{tikzpicture}[thick]
\def\y{.9}
\def\x{1.5}
\def\b{10}
\tikzset{every state/.style = {minimum size =22}}
\node[state] (i) at (0*\x,0) {$q_1$};
\node[state] (m) at (2*\x, 0) {$q_2$};
\node[state,fill=lightgray] (fin) at (4*\x,0) {$f$};

\path[-stealth]
(-.75,0) edge (i)
(i) edge[loop above] node[] {$\varepsilon, X\mid Xa$} ()
 (i) edge node[below] {$\varepsilon,X\mid X$} (m)
 (m) edge[] node[below] {$b,X\mid X$} (fin)
 (m) edge[loop above] node[] {$a, a\mid \epsilon$} ()
;
\end{tikzpicture}
\caption{A PDA recognising $a^* b$ that is weakly compositional but not \hd. Grey states are final, and $X$ is an arbitrary stack
symbol.}
\end{figure}

It is clearly not an \gfgpda, as the resolver would have to predict the number of $a$'s to be seen.
However, for all finitely branching $\Sigma$-arenas $A$, if Player~$2$ wins the game $(A,L(\aut))$, she also wins $G(A,\aut)$. Indeed, if Player~$2$ wins $(A,L(\aut))$ with a strategy $\strat$, then by K\"onig's lemma, there is some bound $n$ on the number of $a$'s in any play that is consistent with $\strat$. Then, in $G(A,\aut)$, Player~$2$ wins by first pushing $n$ letters onto the stack in $\aut$, then playing $\strat$ in the $A$-component. Since $\strat$ guarantees that there are at most $n$ occurrences of $a$ in the play of $L(\aut)$, this strategy is winning for Player~$2$.
\end{proof}

In contrast, \gfgpda~are exactly the \pda~that preserve winners when composed with infinitely branching arenas.

\begin{thm}
A \pda\ is history-deterministic if and only if it is compositional.
\end{thm}

\begin{proof}
One direction of the equivalence was shown in Lemma~\ref{lem:hdimpliescomposition}: every \gfgpda\ is compositional.

For the other direction, we show that from a \pda\ that is not history-deterministic, we can construct an arena that witnesses the failure of compositionality. 
To characterise automata that are not history-deterministic, we rely on the so-called \emph{letter game}, a game-based characterisation of history-determinism due to Henzinger and Piterman~\cite{HP06}.

Given a \pda~$\aut$, the letter game is played by two players, called Challenger and Resolver, and $\aut$ is history-deterministic if and only if Resolver wins the letter game induced by $\aut$. 
The game is played in rounds: in each round $i$, Challenger chooses a letter $a_i$ and then Resolver responds with $\rho_i$, the concatenation of a (potentially empty) sequence of $\varepsilon$-transitions and an $a_i$-transition. Resolver's winning condition is a safety condition: before each round $i$, either the word $w=a_0a_1\cdots a_{i-1}$ built by Challenger so far is not in the language $L(\aut)$ or the sequence $\rho_0\rho_1\cdots \rho_{i-1}$ played by Resolver so far induces an accepting run of $\aut$ over $w$. This game characterises history-determinism, as a winning strategy for Resolver is exactly a resolver witnessing the history-determinism of $\aut$.

We now crucially rely on the determinacy of the letter game, which we can easily show for \pda~over finite words. Indeed, the winning condition for Resolver is a safety condition, in the sense that all losing plays have a finite prefix of which all continuations are losing. The winning condition is therefore Borel, and, by Martin's theorem~\cite{Martin}, determined.
This implies that if $\aut$ is not history-deterministic, then Challenger has a winning strategy ${\strat_C}$ in the letter game.

We now build an arena~$A_{\strat_C}$, based on ${\strat_C}$, such that Player~$2$ wins $(A_{\strat_C}, L(\aut))$, but loses the composition game $G(A_{\strat_C},\aut)$, thus witnessing that $\aut$ is not compositional.

The arena $A_{\strat_C}$ emulates the strategy tree for ${\strat_C}$, in the sense that all the plays of the arena are plays that are consistent with the strategy ${\strat_C}$.
Furthermore, the branching in the arena corresponds to the branching of the strategy tree of ${\strat_C}$, which represents Resolver's choices in the letter game. 
All positions of $A_{\strat_C}$ belong to Player~$1$.

Note that a play in the letter game has the form~$a_0 \rho_0 a_1 \rho_1 a_2 \rho_2 \cdots$ where each $a_i$ is a letter in $\Sigma$ and each $\rho_i$ is a (possibly empty) sequence of $\varepsilon$-transitions concatenated with an $a_i$-transition.
Furthermore, as Resolver picks the $\rho_i$, it is her turn to make a move after a prefix of the form~$a_0 \rho_0 a_1 \cdots \rho_{i-1} a_i$ for $i \ge 0$. 
Now, the positions of $A_{\strat_C}$ are the empty play prefix, which is the initial position, and all play prefixes that are consistent with ${\strat_C}$ and end with a letter, i.e. it is Resolver's turn.
There is an $a$-labelled edge from the initial position $\epsilon$ to the position~$a$, where $a$ is the letter picked by ${\strat_C}$ in round zero. 
Furthermore, for every position of the form~$a_0 \rho_0 a_1 \cdots \rho_{i-2} a_{i-1}\rho_{i-1} a_i$ for $i \ge 1$, there is an $a_{i}$-labelled edge from $a_0 \rho_0 a_1 \cdots \rho_{i-2} a_{i-1}$ to $a_0 \rho_0 a_1 \cdots \rho_{i-2} a_{i-1}\rho_{i-1} a_i$.

Consider a position~$a_0 \rho_0 a_1 \cdots \rho_{i-1} a_i$ for $i \ge 1$. 
If $w = a_0 \cdots a_{i-1}$ is in $L(\aut)$, but $\rho_0 \cdots \rho_{i-1}$ is not an accepting run processing $w$, then we replace the label~$a_i$ of the edge from $a_0 \rho_0 a_1 \cdots \rho_{i-2} a_{i-1}$ to $a_0 \rho_0 a_1 \cdots \rho_{i-2} a_{i-1}\rho_{i-1} a_i$ by $\sink$ and remove all outgoing transitions of $a_0 \rho_0 a_1 \cdots \rho_{i-2} a_{i-1}\rho_{i-1} a_i$.

Note that $A_{\strat_C}$ may be infinitely branching, if Resolver can pick arbitrarily long sequences of $\epsilon$-transitions (e.g. recall the \pda\ presented in the proof of Lemma~\ref{lem:composition-finite}).

Recall that ${\strat_C}$ is a winning strategy for Challenger, i.e. during every play that is consistent with it, there is a round at which the word picked by Challenger is in $L(\aut)$, but Resolver has not picked an accepting run on that word.
Thus, by construction, there is no infinite play in $A_{\strat_C}$ and every maximal play ends in a terminal position and is labelled by a word in $L(\aut)\cdot \sink$.
Hence, Player~$2$ wins $(A_{\strat_C}, L(\aut))$.

Now, towards a contradiction, assume that she also wins $G(A_{\strat_C}, \aut)$, say with a strategy~$\strat_2$. We argue that Resolver can mimic $\strat_2$ in the letter game, and thereby win against ${\strat_C}$, a contradiction. 
Indeed, $\strat_2$ only chooses transitions in $\aut$, since all positions in $A_{\strat_C}$ belong to Player~$1$. 
However, note that $\strat_2$ in $G(A_{\strat_C}, \aut)$ can base its decisions on the sequence of transitions picked by Player~$1$ during the play. 
By only considering those plays where Player~$1$'s choices mimic the choices made by $\strat_2$, which is possible due to the construction of $A_{\strat_C}$, we eliminate this additional information.

So, let us define the strategy for the letter game mimicking $\strat_2$ inductively:
In round zero, Challenger picks some $a$ according to $\strat_C$ in the letter game. 
Now, let $\rho_0$ be the sequence of transitions picked by $\strat_2$ in $G(A_{\strat_C},\aut)$ in reaction to the initial move, which leads via an $a$-labelled edge from $\epsilon$ to $a$.
Then, Resolver picks $\rho_0$ during the first round of the letter game. 

Now, assume we are in round~$i>0$ of the letter game, Challenger has picked $a_0 \cdots a_{i}$ so far, Resolver has answered with $\rho_0 \cdots \rho_{i-1}$, and now has to pick $\rho_i$.
Consider the position~$a_0 \rho_0 \cdots \rho_{i-1}a_i$ of $A_{\strat_C}$. 
As $A_{\strat_C}$ is a tree, there is a unique play prefix ending in that position.
By construction of our strategy, this play prefix is consistent with $\strat_2$.
Now, let $\rho_i$ be the sequence of transitions picked by $\strat_2$ for that play prefix.
Then, Resolver picks $\rho_i$ during the $i^{\text{th}}$ round of the letter game.

Recall that $\strat_2$ is a winning strategy for Player~$2$ in $G(A_{\strat_C},\aut)$.
So, for every maximal play in $A_{\strat_C}$, which by construction is labelled by a word in $L(\aut)\cdot\sink$, $\strat_2$ constructs an accepting run of $\aut$ processing that word.
Thus, the mimicking strategy in the letter game is winning for Resolver.
Thus, we have derived the desired contradiction.

We conclude that if $\aut$ is not history-deterministic, then we can build, from a winning strategy for Challenger in the letter game, an arena such that Player~$2$ wins in the arena, but not in the composition with $\aut$, witnessing that $\aut$ is not compositional.
\end{proof}

Observe that infinite branching is key for this proof to work, as exhibited by the counter-example in the proof of Lemma~\ref{lem:composition-finite}. Another ingredient we use is the determinacy of the letter game. While this is easy enough to obtain for \pda~ over finite words, for which the letter game is a safety game (albeit with a potentially complex safe region), this becomes trickier for \pda~over infinite words. Indeed, then the letter game's winning condition for Resolver is no longer a safety condition and it involves the negation of an $\omega$-context-free language. The determinacy of games with a winning condition given by a B\"uchi one-counter automaton is a large cardinal assumption~\cite{Fin13}, so obtaining determinacy for these letter games on $\omega$-\pda~is likely to be challenging.

\section{Solving Games and Universality}
\label{sec_games.tex}
One of the motivations for \HD automata is that solving games with winning conditions given by \HD automata is easier than for nondeterministic automata. This makes them appealing for applications such as the synthesis of reactive systems, which can be modelled as a game between an antagonistic environment and the system. Solving games is undecidable for PDA in general~\cite{DBLP:journals/tcs/Finkel01a}, both over finite and infinite words, while for DPDA~\cite{DBLP:journals/iandc/Walukiewicz01} and $\omega$-\HD-PDA, it is $\exptime$-complete~\cite{LZ22}. As a corollary, universality is also decidable for $\omega$-\HD-PDA, while it is undecidable for PDA, both over finite and infinite words~\cite{Hopcroft}.

Here, we consider Gale-Stewart games~\cite{GaleStewart53}, abstract games induced by a language in which two players alternately pick letters, thereby constructing an infinite word. One player aims to construct a word that is in the language while the other aims to construct one that is not in the language. We use Gale-Stewart games rather than the arena-based games from the previous section since they are more convenient to work with. 
This is not a restriction, since arena-based games with finite arenas and Gale-Stewart games, both with winning conditions recognised by PDA, are effectively equivalent: a finite arena can be simulated by the PDA recognising the winning condition while a Gale-Stewart game can be seen as an arena-based game with a trivial arena.

Formally, given a language $L\subseteq (\Sigma_1\times \Sigma_2)^*$ of sequences of letter pairs, the game $G(L)$ is played between Player~1 and Player~2 in rounds $i=0,1,\ldots$ as follows: At each round~$i$, Player~1 plays a letter $a_i\in\Sigma_1$ and Player~2 answers with a letter $b_i\in \Sigma_2$. A play of $G(L)$ is an infinite word~$\binom{a_0}{b_0}\binom{a_1}{b_1}\cdots$ and Player~2 wins such a play if and only if each of its prefixes is in the language $L$.
 A strategy for Player~$2$ is a mapping from $\Sigma_1^+$ to $\Sigma_2$ that gives for each prefix played by Player~1 the next letter to play. A play is consistent with a strategy $\sigma$ if for each $i$, $b_i=\sigma(a_0a_1\dots a_i)$. 
 Player~2 wins $G(L)$ if she has a strategy such that all plays that are consistent with it are winning for Player~2. Observe that Player~2 loses whenever the projection of $L$ onto its first component is not universal.
Finally, universality reduces to solving these games: $\aut$ is universal if and only if Player~2 wins $G(L)$ for $L=\{\binom{w(0)}{\#} \cdots \binom{w(n)}{\#}\mid w(0) \cdots w(n)\in L(\aut)\}$.

We now argue that solving games for \HD-PDA easily reduces to the case of $\omega$-\HD-PDA, which are just \HD-PDA over infinite words, where acceptance is not determined by final state, since runs are infinite, but rather by the states or transitions visited infinitely often. Here, we only need safety $\omega$-\HD-PDA, in which every infinite run is accepting (i.e. rejection is implemented via missing transitions). The infinite Gale-Stewart game over a language $L$ of infinite words, also denoted by $G(L)$, is as above, except that victory is determined by whether the infinite word built along the play is in $L$.

\begin{lem}
\label{lemma_gamereduction}
Given an \HD-PDA $\aut$, there is a safety $\omega$-\HD-PDA $\aut'$ no larger than $\aut$ such that Player~2 wins $G(L(\aut))$ if and only if she wins $G(L(\aut'))$.
\end{lem}

\begin{proof}
Let $\aut'$ be the PDA obtained from $\aut$ by  removing all transitions~$(q,X,a,q',\gamma)$ of $\aut$ with $a \in \Sigma$ and with non-final~$q'$.
Note that all $\epsilon$-transitions are preserved.

With a safety condition, in which every infinite run is accepting, $\aut'$ recognises exactly those infinite words whose prefixes are all accepted by $\aut$. 
Hence, the games $G(L(\aut))$ and $G(L(\aut'))$ have the same winning player.
Note that the correctness of this construction crucially relies on our definition of \HD-PDA, which requires a run on a finite word to end as soon as the last letter is processed.
Then, the word is accepted if and only if the state reached by processing this last letter is final.

Finally, since $\aut$ is \HD, so is $\aut'$. Consider an infinite input in $L(\aut')$. Then, every prefix~$w$ has an accepting run of $\aut$ induced by its resolver, which implies that the last transition of this run (which processes the last letter of $w$) is not one of those that are removed to obtain~$\aut'$. Now, an induction shows that the same resolver works for $\aut'$ as well, relying on the fact that if $w$ and $w'$ with $\size{w}<\size{w'}$ are two such prefixes, then the resolver-induced run of $\aut$ on $w$ is a prefix of the resolver-induced run of $\aut$ on $w'$.
\end{proof}

Our main results of this section are now direct consequences of the corresponding results on $\omega$-words~\cite[Theorem~4.2 \& Corollary~4.3]{LZ22}.

\begin{cor}
Given an \HD-PDA $\aut$, deciding whether $L(\aut)=\Sigma^*$ and whether Player~2 wins $G(L(\aut))$ are both in $\exptime$.
\end{cor}

Determining whether, for a given \HD-PDA $\aut$, Player~2 wins $G(L(\aut))$ is in fact $\exptime$-complete, which can be deduced from Walukiewicz's result showing that solving games with $\omega$-\dcfl winning conditions is $\exptime$-hard~\cite[Proposition~31 \& Theorem~26]{DBLP:journals/iandc/Walukiewicz01}.
Whether universality of \HD-PDA is also $\exptime$-hard is left open, as it is for \HD-PDA over $\omega$-words~\cite{LZ22}.


\section{Resolvers}
\label{sec_resolvers}

The definition of a resolver does not put any restrictions on its complexity. In this section we study the complexity of the resolvers that \HD-PDA need. On one hand, we show that resolvers can always be chosen to be \textit{positional}, that is, dependent only on the current configuration of the automaton (which, for pushdown automata, includes the current state and the current stack content). Note that this is not the case for $\omega$-regular automata\footnote{A positional resolver for $\omega$-regular automata implies determinisability by pruning, and we know that this is not always possible~\cite{BKS17}.}, let alone $\omega$-\HD-PDA. On the other hand, we show that they are not always implementable by pushdown transducers.

More formally, a resolver $r$ is positional, if for any two sequences $\rho$ and $\rho'$ of transitions inducing runs ending in the same configuration, $r(\rho, a)=r(\rho',a)$ for all $a\in \Sigma$.

\begin{lem}
\label{lemma_posresolver}
Every \HD-PDA has a positional resolver.
\end{lem}

\begin{proof}
Let $r'$ be a (not necessarily positional) resolver for $\aut$. We define a resolver~$r$ such that for each configuration and input letter, it makes a choice consistent with $r'$ for some input leading to this configuration.
In other words, for every configuration~$c$ reachable by a run induced by $r'$, let $\rho_c$ be an input to $r'$ inducing a run ending in $c$. 
Then, we define $r(\rho, a) = r'(\rho_c, a)$, where $c$ is the last configuration of the run induced by $\rho$.
Note that every configuration reachable by a run induced by $r$ is also reachable by a run induced by $r'$.

We claim that $r$, which is positional by definition, is a resolver.
Towards a contradiction, assume that this is not the case, i.e. there is a word $w\in L(\aut)$ such that the run $\rho$ induced by $r$ is rejecting. Since this run is finite and $w \in L(\aut)$, there is some last configuration $c$ along the run $\rho$ from which the rest of the word, say $u$, is accepted\footnote{Observe that this is no longer true over infinite words as an infinite run can stay within configurations from where an accepting run exists without being itself accepting. In fact, the lemma does not even hold for coBüchi automata~\cite{KS15} as the existence of positional resolvers implies determinisability by pruning.} (by some other run of $\aut$ having the same prefix as $\rho$ up to configuration~$c$). Let $\tau$ be the next transition along $\rho$ from $c$. Since $r$ chose $\tau$, the resolver $r'$ also chooses $\tau$ after some history leading to $c$, over some word $v$. Since $u$ is accepted from $c$, the word $vu$ is in $L(\aut)$; since $r'$ is a resolver, there is an accepting run over $u$ from $c$ starting with $\tau$, contradicting that $c$ is the last position on $\rho$ from where the rest of the word could be accepted. 
\end{proof}

Contrary to the case of finite and $\omega$-regular automata, since \HD-PDA  have an infinite configuration space, the existence of positional resolvers does not imply determinisability. 
On the other hand, if an \HD-PDA has a resolver which only depends on the mode of the current configuration, then it is \emph{determinisable by pruning}, as transitions that are not used by the resolver can be removed to obtain a deterministic automaton.
However, not all \HD-PDA are determinisable by pruning, e.g. the \HD-PDA for the languages~$C_n$ used to prove Theorem~\ref{theorem:succinctness}.

As even positional resolvers for \HD-PDA are infinite objects mapping configurations to transitions, they are not necessarily finitely representable.
We now turn our attention to whether \HD-PDA always have finitely representable resolvers, i.e., to what kind of memory a \emph{machine} needs in order to implement a resolver for an \HD-PDA.
First, we introduce transducers as a way to implement a resolver. A transducer is an automaton with outputs instead of acceptance, i.e. it computes a function from input sequences to outputs. A pushdown resolver is a pushdown transducer that implements a resolver.

Note that a resolver has to pick enabled transitions in order to induce accepting runs for all inputs in the language.
To do so, it needs access to the mode of the last configuration. 
However, to keep track of this information on its own, the pushdown resolver would need to simulate the stack of the \HD-PDA it controls.
This severely limits the ability of the pushdown resolver to implement computations on its own stack.
Thus, we give a pushdown resolver access to the current mode of the \HD-PDA via its output function, thereby freeing its own stack to implement further functionalities.

Formally, a pushdown transducer (PDT for short) $\taut= (\daut, \lambda)$ consists of a DPDA $\daut$ augmented with an output function $\lambda : Q^\daut \rightarrow \Theta$ mapping the states $Q^\daut$ of $\daut$ to an output alphabet $\Theta$.
The input alphabet of $\taut$ is the input alphabet of $\daut$.

Then, given a PDA $\aut=(Q, \Sigma, \Gamma, q_\initmark, \Delta, F )$, a pushdown resolver for $\aut$ consists of a pushdown transducer $\taut= (\daut,\lambda)$ with input alphabet $\Delta$ and output alphabet $Q\times \Gammabot \times \Sigma\rightarrow \Delta$ such that
the function $r_\taut$, defined as follows, is a resolver for $\aut$:
$ r_\taut(\tau_0\dots \tau_k,a)= \lambda(q_\taut)(q_\aut,X,a)$
where
\begin{itemize}
    \item $q_\taut$ is the state of the last configuration of the longest run of $\daut$ on $\tau_0\dots \tau_k$ (recall that while $\daut$ is deterministic, it may have several runs on an input which differ on trailing $\epsilon$-transitions);
    \item $(q_\aut, X)$ is the mode of the last configuration of the run of $\aut$ induced by $\tau_0\dots\tau_k$.
\end{itemize}
In other words, a transducer implements a resolver by processing the run so far, and then uses the output of the state reached and the state and top stack symbol of the \HD-PDA to determine the next transition in the \HD-PDA.

\begin{figure}
\centering
\begin{tikzpicture}[thick]
\def\y{1.1}
\def\x{1.35}
\tikzset{every state/.style = {minimum size =22}}
\node[state] (i) at (0*\x,0) {$q_1$};
\node[state] (ii) at (2*\x, 0) {$q_2$};
\node[state] (iii) at (4*\x,0) {$q_3$};
\node[state] (iv) at (4*\x,-2*\y) {$p_1$};
\node[state] (vi) at (6*\x,-2*\y) {$p_2$};
\node[state] (vii) at (8*\x,-2*\y) {$p_3$};
\node[state,fill=lightgray] (viii) at (10*\x,-2*\y) {$f$};

\path[-stealth]
(-.75,0) edge (i)
(i) edge[loop above] node[] {$a, X \mid Xa $} ()
(ii) edge[loop above] node[] {$a, X \mid Xa $} ()
(iii) edge[loop above] node[] {$a, X \mid Xa $} ()
(iv) edge[loop below] node[] {$\epsilon, a \mid \epsilon $} ()
(vi) edge[loop below] node[] {$\epsilon, a \mid \epsilon $} ()
(vii) edge[loop below] node[] {$b,a\mid \epsilon$} ()
(i) edge node[above] {$\$,X \mid X\$ $} (ii)
(ii) edge node[above] {$\$,X\mid X\$ $} (iii)
(iii) edge[] node[left] {$\$, X \mid X$} (iv)
(iii) edge[bend left =20] node[right,near start,yshift=.2cm] {$\$, X \mid X$} (vii)
(iii) edge[bend left=0] node[right,xshift=.1cm] {$\$, X \mid X$} (vi)
(iv) edge node[above,near start,xshift=.2cm] {$\epsilon, \$ \mid \epsilon$} (vi)
(vi) edge node[above, near start,xshift=.2cm] {$\epsilon, \$ \mid \epsilon$} (vii)
(vii) edge node[above] {$\$,X \mid X\$ $} (viii)
;
\end{tikzpicture}
    \caption{The PDA $\aut_{\lmax}$ for $\lmax$. Grey states are final, and $X$ is an arbitrary stack symbol.}
    \label{fig:lmax}
\end{figure}

We now give an example of an \HD-PDA which does not have a pushdown resolver. The language in question is the language~$
\lmax=\{a^i\$a^j\$a^k\$b^l\$\mid l\leqslant\max(i,j,k)\}$.
 Compare this to the language $B_2$ in Section~\ref{sec_expressivity} which \emph{does} have a pushdown resolver.
Let $\aut_{\lmax}$ be the automaton in Figure \ref{fig:lmax}, which works analogously to the automaton for $B_2$ in Figure~\ref{fig:gfgpdaexample}.

This automaton recognises $\lmax$: for a run to end in the final state, the stack, and therefore the input, must have had an $a$-block longer than or equal to the final $b$-block; conversely, if the $b$-block is shorter than or equal to some $a$-block, the automaton can nondeterministically pop the blocks on top of the longest $a$-block off the stack before processing the $b$-block. Furthermore, this automaton is \HD: the nondeterminism can be resolved by removing from the stack all blocks until the \textit{longest} $a$-block is at the top of the stack, and this choice can be made once the third $\$$ is processed.

We now argue that this \HD-PDA needs more than a pushdown resolver.
The reason is that a pushdown resolver needs to be able to determine which of the three blocks is the longest while processing a prefix of the form~$a^*\$a^*\$a^*$. However, at least one of the languages induced by these three choices is not context-free, yielding the desired contradiction. 

\begin{lem}
\label{nopdtresolver}
The \HD-PDA $\aut_{\lmax}$ has no pushdown resolver.
\end{lem}

\begin{proof}
Towards a contradiction, assume that there is a pushdown resolver $r$ for $\aut_{\lmax}$,
implemented by a PDT $\taut= (\daut, \lambda)$.

We construct, for each $i\in \{1,2,3\}$, a PDA~$\daut_i$ that recognises the language of words~$w \in a^+ \$ a^+ \$ a^+$ such that $\taut$ chooses the transition of $\aut_{\lmax}$ going from $q_3$ to $p_i$ when constructing a run on $w\$$.
Then, we show that at least one of the languages $L(\daut_i)$ is not context-free, thereby obtaining the desired contradiction. 

To this end, we first restrict the possible inputs $\daut$ can process.
Let $\tau_i$ for $i \in \set{1,2,3}$ be the self-loop processing $a$ on state~$q_i$ of $\aut_{\lmax}$ and let $\tau_i'$ for $i \in\set{1,2}$ be the $\$$-transition from $q_i$ to $q_{i+1}$.
Now, let $\auta$ be a DFA for the regular language~$\tau_1^+ \tau_1' \tau_2^+ \tau_2' \tau_3^+$ and let $\daut'$ be obtained by taking the product of $\daut$ with $\auta$. 
Note that $\daut'$ does not process any $\epsilon$-transitions nor $b$-transitions of $\aut_{\lmax}$.

Now the PDA~$\daut_i$ for $i\in\set{1,2,3}$ is obtained from $\daut'$ by modifying each transition of $\daut'$ processing a transition~$\tau$ of $\aut_{\lmax}$ to now process $\ell(\tau)\in \{a,\$\}$, and final states are defined as follows:
Recall that states of $\daut'$ are of the form~$(q,q')$ where $q$ is a state of $\daut$ and $q'$ is a state of $\auta$. 
Such a state is final in $\daut_i$ if $q'$ is final in $\auta$, i.e. a word of the form $a^+\$a^+\$a^+$ has been processed, and $\lambda(q)(q_3,a,\$)=(q_3,\$, a, p_i, a)$, i.e. $\taut$ chooses the transition of $\aut_{\lmax}$ going from $q_3$ to $p_i$.

The language $a^+\$a^+\$a^+$
is partitioned among the languages recognized by~$\daut_1$, $\daut_2$ and~$\daut_3$,
which implies that there is at least one $i\in \{1,2,3\}$
such that $\daut_i$ accepts $a^m\$a^m\$a^m$ for infinitely many $m$'s.
Let us fix such an index $i$ for the rest of the proof.
Note that for every $m_1,m_2,m_3 \in \mathbb{N}$, the resolver $\daut$
needs to be able to prolong the run selected over the input $a^{m_1}\$a^{m_2}\$a^{m_3}\$$
into an accepting run if the suffix $b^{\max(m_1,m_2,m_3)}$ is encountered.
This is only possible from the states $q_j$ satisfying $\max(m_1,m_2,m_3)= m_j$,
which implies that every word $a^{m_1}\$a^{m_2}\$a^{m_3}$
accepted by $\daut_i$ satisfies $\max(m_1,m_2,m_3)= m_i$.

To reach a contradiction, we now argue that this $\daut_i$ recognises a language that is not context-free. Indeed, if it were, then by applying the pumping lemma for context-free languages, there would be a large enough $m$ such that the word $a^m\$a^m\$a^m\in L(\daut_i)$ could be decomposed as $uvwyz$ such that $|vy|\geqslant 1$ and $uv^n w y^n z$ is in the language of $\daut_i$ for all $n\geqslant 0$. In this decomposition, $v$ and $y$ must be $\$$-free. Then, if either $v$ or $y$ occurs in the $i^{\text{th}}$ block and is non-empty, by setting $n=0$ we obtain a contradiction as the $i^{\text{th}}$ block is no longer the longest. Otherwise, we obtain a similar contradiction by setting $n=2$. In either case, this shows that $\taut$ is not a pushdown resolver for $\aut$.
\end{proof}

Another restricted class of resolvers are finite-state resolvers, which can be seen as pushdown resolvers that do not use their stack. 
Similarly to the case of $\omega$-\HD-PDA~\cite[Theorem~7.2]{LZ22}, the product of an \HD-PDA and a finite-state resolver yields a DPDA for the same language.

\begin{rem}
Every \HD-PDA with a finite-state resolver is determinisable.
\end{rem}

Note that the converse does not hold. For example, consider the regular, and therefore deterministic context-free, language~$L_{10}$ of words~$w\#$ with $w \in \set{a,b}^*$ with infix $a^{10}$. An \HD-PDA~$\aut_{10}$ recognising $L_{10}$ can be constructed as follows: $\aut_{10}$ pushes its input  onto its stack until processing the first $\#$. Before processing this letter, $\aut_{10}$ uses $\epsilon$-transitions to empty the stack again. While doing so, it can nondeterministically guess whether the next $10$ letters removed from the stack are all $a$'s. If yes, it accepts; in all other cases (in particular if the input word does not end with the first $\#$ or the infix~$a^{10}$ is not encountered on the stack) it rejects.
This automaton is history-deterministic, as a resolver has access to the whole prefix before the first $\#$ when searching for $a^{10}$ while emptying the stack. This is sufficient to resolve the nondeterminism.
On the other hand, there is no finite-state resolver for $\aut_{10}$, as resolving the nondeterminism, intuitively, requires to keep track of the whole prefix before the first $\#$ (recall that a finite-state resolver only has access to the topmost stack symbol).

We consider another model of pushdown resolver, namely one that does not only have access to the mode of the \HD-PDA, but can check the full stack for regular properties. 
Thus, we consider a more general model where the transducer can use information about the whole stack when determining the next transition. 
More precisely, we consider a regular abstraction of the possible stack contents by fixing a DFA running over the stack and allowing the transducer to base its decision on the state reached by the DFA as well.

Then, given a PDA $\aut=(Q, \Sigma, \Gamma, q_\initmark, \Delta, F )$, a pushdown resolver with regular stack access $\taut= (\daut, \auta, \lambda)$ consists of a DPDA~$\aut$ with input alphabet~$\Delta$, a DFA~$\auta$ over $\Gammabot$ with state set~$Q^\auta$, and an output function~$\lambda$ with output alphabet~$Q\times Q^\auta \times \Sigma\rightarrow \Delta$ such that
the function $r_\taut$ defined as follows, is a resolver for $\aut$:
\[ r_\taut(\tau_0\dots \tau_k,a)= \lambda(q_\taut)(q_\aut,q_\auta,a)\]
where
\begin{itemize}
    \item $q_\taut$ is the state of the last configuration of the longest run of $\daut$ on $\tau_0\dots \tau_k$ (recall that while $\daut$ is deterministic, it may have several runs on an input which differ on trailing $\epsilon$-transitions).
    \item Let $c$ be the last configuration of the run of $\aut$ induced by $\tau_0\dots\tau_k$. Then, $q_\aut$ is the state of $c$ and $q_\auta$ is the state of $\auta$ reached when processing the stack content of $c$.
\end{itemize}

Every pushdown resolver with only access to the current mode is a special case of a pushdown resolver with regular stack access. 
On the other hand, having regular access to the stack is strictly stronger than having just access to the mode. However, by adapting the underlying \HD-PDA, one can show that the \emph{languages} recognised by \HD-PDA with pushdown resolvers does not increase when allowing regular stack access. 

\begin{lem}
Every \HD-PDA with a pushdown resolver with regular stack access can be turned into an equivalent \HD-PDA with a pushdown resolver.
\end{lem}

\begin{proof}
Let $\aut=(Q, \Sigma, \Gamma, q_\initmark, \Delta, F )$ be an \HD-PDA and let $(\daut, \auta, \lambda)$ be a pushdown resolver with stack access for $\aut$.
We keep track of the state $\auta$ reaches on the current stack as in the proof of Lemma~\ref{lemma_eow}:
If a stack content~$\bot(X_1, q_1) \cdots (X_s, q_s)$ is reached, then $q_j$ is the unique state of $\aut$ reached when processing $\bot X_1 \cdots X_j$.
Now, it is straightforward to turn $(\daut, \auta, \lambda)$ into a pushdown resolver for $\aut$ that has only access to the top stack symbol.
\end{proof}

Finally, for \HD-VPA, the situation is again  much better. Using the one-token game discussed in the proof of Theorem~\ref{thm:vpagfgnesscomplexity} and known results~\cite{DBLP:conf/fsttcs/LodingMS04} about VPA games having VPA strategies, we obtain our final theorem.

\begin{thm}
\label{theorem_visiblyresolver}
Every \HD-VPA has a (visibly) pushdown resolver that can be computed in exponential time.
\end{thm}

\begin{proof}
Recall the one-token game that characterises \HDness of a VPA~$\aut$ from the proof of Theorem \ref{thm:vpagfgnesscomplexity}.
We argued there that the one-token game is a safety game played on the configuration graph of a deterministic VPA whose size is polynomial in the size of $\aut$. 
Therefore, Player~$2$, when she wins, has a winning strategy that can be implemented by a visibly pushdown transducer, which can be computed in exponential time~\cite{WF,DBLP:journals/iandc/Walukiewicz01}.
The proof of Theorem~\ref{thm:vpagfgnesscomplexity} shows that such a strategy for Player~$2$ in the one-token game for VPA induces a resolver. We conclude that \HD-VPA have visibly pushdown resolvers computable in exponential time.
\end{proof}

Note that in the above proof, we could also compute directly a strategy to the letter game, rather than the one-token game, by using the closure of VPA under complementation and union to encode the letter game as a safety game on the configuration graph of a VPA. However, the resulting VPA would be of exponential, rather than polynomial, size.


\section{Conclusion}
\label{section_conc}
We have continued the study of history-deterministic pushdown automata initiated by Lehtinen and Zimmermann~\cite{LZ22}, focusing here on expressiveness and succinctness. In particular, we have shown that \HD-PDA are not only more expressive than DPDA (as had already been shown for the case of infinite words), but also more succinct than DPDA:
We have introduced the first techniques for using \HD nondeterminism to succinctly represent languages that do not depend on the coBüchi condition.
Similarly, for the case of VPA, for which deterministic and nondeterministic automata are equally expressive, we proved a (tight) exponential gap in succinctness.

Solving games and universality are decidable for \HD-PDA, but \HDness is undecidable and \HD-PDA have limited closure properties.
On the other hand, \HDness for VPA is decidable and they inherit the closure properties of VPA, e.g. union, intersection and complementation, making \HD-VPA an exciting class of pushdown automata.
Finally, we have studied the complexity of resolvers for \HD-PDA, showing that positional ones always suffice, but that they are not always implementable by pushdown transducers.
Again, \HD-VPA are better-behaved, as they always have a resolver implementable by a VPA.

Let us conclude by mentioning some open problems raised by our work.
\begin{itemize}

    \item Is universality for \HD-PDA $\exptime$-hard? Note that the analogous problem for \HD-PDA over $\omega$-words is also open~\cite{LZ22}.

    \item It is known that the succinctness gap between PDA and DPDA is noncomputable~\cite{Hartmanis80,Valiant76}, i.e. there is no computable function~$f$ such that any PDA of size~$n$ that has some equivalent DPDA also has an equivalent DPDA of size~$f(n)$. Due to our hierarchy results, at least one of the succinctness gaps between PDA and \HD-PDA and between \HD-PDA and DPDA has to be uncomputable, possibly both.
    
    \item We have shown that some \HD-PDA do not have pushdown resolvers, but there are no known upper bounds on the complexity of resolvers.
    It is even open whether every \HD-PDA has at least a computable resolver. 
    
    \item On the level of languages, it is open whether every language in \gfgcfl has an \HD-PDA recognising it with a resolver implementable by a pushdown transducer.
    
    \item We have shown that \HDness is undecidable, both for PDA and for context-free languages. Is it decidable whether a given \HD-PDA has an equivalent DPDA?
    
    \item Equivalence of DPDA is famously decidable~\cite{DBLP:journals/tcs/Senizergues01} while it is undecidable for PDA~\cite{Hopcroft}. Is equivalence of \HD-PDA decidable?
    
    \item Does every \HD-PDA that is equivalent to a DPDA have a finite-state resolver with regular stack access?
    
    \item There is a plethora of fragments of context-free languages one can compare \gfgcfl to, let us just mention a few interesting ones: Height-deterministic context-free languages~\cite{NS07}, context-free languages with bounded nondeterminism~\cite{Her97} and preorder typeable visibly pushdown languages~\cite{KMV07}. 
    
    \item We have shown that history-determinism and good-for-gameness coincide for PDA over finite words; it remains open whether this is also the case for PDA over infinite words.
\end{itemize}



\bibliographystyle{alphaurl}
\bibliography{content/biblio.bib}



\end{document}